\documentclass[11pt,letterpaper]{article}

\usepackage[margin=1in]{geometry}
\usepackage[affil-it]{authblk}
\usepackage{amsmath,amssymb,amsthm,mathtools}
\usepackage{graphicx}
\usepackage{microtype}
\usepackage{booktabs,array,enumitem}
\usepackage{algorithm,algpseudocode}
\usepackage{placeins,needspace}
\usepackage{xcolor}
\usepackage{aliascnt}
\usepackage[colorlinks=true,linkcolor=blue!55!black,citecolor=green!45!black,urlcolor=blue!65!black]{hyperref}
\usepackage[capitalize,noabbrev]{cleveref}
\crefname{appendix}{appendix}{appendices}
\Crefname{appendix}{Appendix}{Appendices}
\makeatletter
\providecommand*{\theHALG@line}{}
\renewcommand*{\theHALG@line}{\thealgorithm.\arabic{ALG@line}}
\makeatother
\usepackage{url}

\allowdisplaybreaks
\setlist[itemize]{leftmargin=1.6em,itemsep=0.25em,topsep=0.4em}
\setlist[enumerate]{leftmargin=1.8em,itemsep=0.25em,topsep=0.4em}

\newtheorem{theorem}{Theorem}[section]
\newaliascnt{proposition}{theorem}
\newtheorem{proposition}[proposition]{Proposition}
\aliascntresetthe{proposition}
\newaliascnt{lemma}{theorem}
\newtheorem{lemma}[lemma]{Lemma}
\aliascntresetthe{lemma}
\newaliascnt{corollary}{theorem}

\aliascntresetthe{corollary}
\theoremstyle{definition}
\newaliascnt{definition}{theorem}
\newtheorem{definition}[definition]{Definition}
\aliascntresetthe{definition}
\newaliascnt{problem}{theorem}
\newtheorem{problem}[problem]{Problem}
\aliascntresetthe{problem}
\newaliascnt{construction}{theorem}

\aliascntresetthe{construction}
\crefname{construction}{construction}{constructions}
\Crefname{construction}{Construction}{Constructions}
\crefname{problem}{problem}{problems}
\Crefname{problem}{Problem}{Problems}
\theoremstyle{remark}
\newaliascnt{remark}{theorem}

\aliascntresetthe{remark}

\newcommand{\C}{\mathbb C}
\newcommand{\R}{\mathbb R}
\newcommand{\cH}{\mathcal H}
\newcommand{\cK}{\mathcal K}
\newcommand{\cM}{\mathcal M}

\newcommand{\ket}[1]{\lvert #1\rangle}
\newcommand{\bra}[1]{\langle #1\rvert}
\newcommand{\proj}[1]{\lvert #1\rangle\!\langle #1\rvert}
\newcommand{\norm}[1]{\lVert #1\rVert}
\newcommand{\ip}[2]{\langle #1,#2\rangle}
\newcommand{\tr}{\operatorname{tr}}
\newcommand{\spanR}{\operatorname{span}_{\R}}
\newcommand{\diag}{\operatorname{diag}}
\newcommand{\Dtr}{D_{\mathrm{tr}}}
\newcommand{\Coarse}{\textsc{PrepareKernelComponent}}
\newcommand{\FilterKernel}{\textsc{FilterKernel}}
\newcommand{\CorrectSolution}{\textsc{CorrectSolution}}
\newcommand{\Solver}{\textsc{Solve}}
\newcommand{\rD}{\mathsf D}
\newcommand{\rG}{\mathsf G}
\newcommand{\rA}{\mathsf A}
\newcommand{\rF}{\mathsf F}
\newcommand{\rC}{\mathsf C}
\newcommand{\rS}{\mathsf S}
\newcommand{\rT}{\mathsf T}
\newcommand{\rB}{\mathsf B}
\newcommand{\rE}{\mathsf E}
\newcommand{\Finite}{\textsc{ImplementTransducer}}
\newcommand{\transduce}[1]{\stackrel{#1}{\rightsquigarrow}}
\newcommand{\DownTransduce}{\scalebox{0.7}{\rotatebox[origin=c]{270}{$\rightsquigarrow$}}}

\title{Simultaneously Query-Optimal Quantum Linear-System Algorithm}
\author[1,2,3]{Minbo Gao\thanks{%
  \href{mailto:gaomb@ios.ac.cn}{\nolinkurl{gaomb@ios.ac.cn}} or
  \href{mailto:gmb17@tsinghua.org.cn}{%
    \nolinkurl{gmb17@tsinghua.org.cn}}.}}
\author[4]{Zhengfeng Ji\thanks{%
  \href{mailto:jizhengfeng@tsinghua.edu.cn}{%
    \nolinkurl{jizhengfeng@tsinghua.edu.cn}}.}}
\author[1,2]{Chenghua Liu\thanks{%
  \href{mailto:liuch.russell@gmail.com}{%
    \nolinkurl{liuch.russell@gmail.com}}.}}
\affil[1]{Institute of Software, Chinese Academy of Sciences.}
\affil[2]{University of Chinese Academy of Sciences.}
\affil[3]{TraverseQuantum Co., Ltd.}
\affil[4]{Tsinghua University.}
\date{}

\begin{document}
\maketitle
\begin{abstract}

Quantum linear-system algorithms access the coefficient matrix through a
block-encoding oracle $U_A$ and the right-hand side through a
state-preparation oracle $U_b$.
The optimal query complexities of these two oracles are known separately:
$O(\kappa\log(1/\varepsilon))$ queries to $U_A$ and
$O(\kappa/s)$ queries to $U_b$, where
$s=\alpha\|A^{-1}b\|$.
Whether a single algorithm can achieve both bounds simultaneously has
remained open.

We resolve this question affirmatively.
Given a constant-factor estimate of $s$, we give a quantum linear-system
algorithm that prepares the normalized solution to error $\varepsilon$ using
simultaneously
$ O\left(\kappa\log\frac1\varepsilon\right)$ queries to $U_A$
and   $O\left({\kappa}/{s}\right)$ queries to $U_b$.
The algorithm separates state preparation from precision refinement, so only
the matrix-query count depends on the target accuracy.

We also prove a matching same-instance worst-case lower bound: for every
quantum linear-system algorithm and every solution-norm scale, there exists a
single Hermitian instance on which that algorithm requires both
$\Omega(\kappa\log(1/\varepsilon))$ matrix queries and
$\Omega(\kappa/s)$ state-preparation queries.
Thus the two oracle complexities can be optimized simultaneously, and the
resulting bounds are jointly optimal.

\end{abstract}

\newpage

\section{Introduction}\label{sec:introduction}
Linear systems are a basic primitive in computation.  They arise throughout
scientific computing, including in differential equations, optimization,
simulation, and data analysis.  Classically, solving a linear system means
computing an explicit description of the solution vector.  For high-dimensional
problems, however, even writing down this vector can be prohibitively expensive.
This makes linear systems a natural target for quantum algorithms, where a
high-dimensional vector can instead be represented as a quantum state.

The quantum linear-systems problem (QLS), introduced by Harrow, Hassidim, and
Lloyd~\cite{HHL09}, asks to prepare a quantum state proportional to $A^{-1}b$,
given quantum access to a matrix $A$ and a vector $b$.  The motivation is that
many applications do not require a classical description of every entry of the
solution.  Instead, the solution may be used as an intermediate state in a
larger quantum computation, or one may only need to estimate a property of it.
Quantum linear-system solvers have since been used as subroutines for
solving ordinary and partial differential equations
\cite{BCOW17,CLO21,Kro23,BC24},
evaluating matrix functions~\cite{TOSU20},
and solving generalized eigenvalue problems~\cite{SL22}.
These and other applications are discussed recently by Low and
Su~\cite{LowSu26}.

Since HHL, a central goal has been to reduce the cost of this subroutine.
Successive works improved its dependence on the condition number $\kappa$ and
the target error $\varepsilon$: variable-time techniques improved the
condition-number dependence~\cite{Ambainis12}, while later polynomial
approximation methods achieved exponentially better dependence on precision
\cite{CKS17}.  More recent algorithms based on block encodings, discrete
adiabatic evolution, and kernel reflection achieve the optimal matrix-query
complexity
\[
    O\!\left(\kappa\log\frac1\varepsilon\right)
\]
in the standard block-encoding model~\cite{Costa22,Dalzell24}.  Thus, for
queries to the matrix, the worst-case dependence on $\kappa$ and $\varepsilon$
is essentially understood.

There is, however, another resource in QLS.  In the block-encoding model,
the matrix is accessed through an oracle $U_A$, while the right-hand side is
provided by a separate state-preparation oracle $U_b$ satisfying
$U_b\ket{0^n}=\ket b$.  These two operations can have very different costs.
For example, $\ket b$ may itself be produced by a costly quantum subroutine,
so repeatedly preparing it can dominate the overall computation; in other
settings, access to $A$ may be the expensive part.  It is therefore natural
to count queries to $U_A$ and $U_b$ separately.

The number of queries to $U_b$ is governed by the norm
$\|A^{-1}b\|$ of the unnormalized solution.  For a normalized matrix with
condition number at most $\kappa$, this quantity lies between $1$ and
$\kappa$.  Recent work of Low and Su~\cite{LowSu26} showed that the optimal
state-preparation complexity is
\[
    O\!\left(\frac{\kappa}{\|A^{-1}b\|}\right)
\]
queries to $U_b$.  Their algorithm, however, uses more than the optimal
number of matrix queries.  Conversely, existing matrix-optimal QLS
algorithms use $O\!\left(\kappa\log\frac1\varepsilon\right)$
queries to $U_A$, but do not simultaneously achieve the optimal
state-preparation complexity.  Thus, the two oracle complexities were known
to be optimal separately, but no single QLS algorithm was known to attain
both optima simultaneously.

This leads to the central question of this work:

\begin{center}
\emph{Can a quantum linear-system algorithm be simultaneously
query-optimal with respect to both input oracles?}
\end{center}

We answer this question affirmatively.  We give a QLS algorithm that uses
$O(\kappa\log(1/\varepsilon))$ queries to $U_A$ and
$O(\kappa/\|A^{-1}b\|)$ queries to $U_b$, simultaneously attaining the
optimal query complexity for both input oracles.  The algorithm is also
gate-efficient, using
\[
    O\!\left(
        \kappa(a+1)\log\frac1\varepsilon
        +\frac{\kappa(1+\log d)}{\|A^{-1}b\|}
    \right)
\]
additional one- and two-qubit gates.  We further prove matching lower
bounds on a single hard instance, establishing simultaneous optimality of
the two query complexities.

\subsection{Main results and comparison with prior work}

We work in the standard block-encoding model.  Let $U_A$ be an exact $(\alpha,a,0)$-block-encoding of $A$, let $U_b$ prepare $\ket b$, and suppose that a constant-factor estimate $\widehat s$ of $s=\alpha\norm{A^{-1}b}$ is available.  Controlled calls and adjoints of both input unitaries are allowed.  The formal access and error conventions are stated in \cref{prob:qlsa}.

\begin{theorem}[Informal version of~\cref{thm:upper}]\label{thm:upper-summary}
Suppose that $\kappa\ge\alpha\norm{A^{-1}}$ and
$0<\varepsilon<1/2$.  There is a quantum algorithm that outputs
$\ket{\widetilde x}$ with probability at least $2/3$ such that
\[
 \left\|\widetilde x-
 \frac{A^{-1}b}{\norm{A^{-1}b}}\right\|\le\varepsilon.
\]
It uses $O(\kappa\log(1/\varepsilon))$ queries to $U_A$,
$O(\kappa/s)$ queries to $U_b$, and
$O(\kappa(a+1)\log(1/\varepsilon)+\kappa(1+\log d)/s)$ extra one- and
two-qubit gates.
\end{theorem}

These bounds are jointly optimal.  Separate lower bounds would still allow
the matrix-hard and vector-hard inputs to be different.  We show instead that,
at every solution-norm scale, one input forces both costs simultaneously.

\begin{theorem}[Informal version of~\cref{thm:lower}]
\label{thm:lower-summary}
Fix sufficiently large $\kappa$, sufficiently small $\varepsilon$, and a
solution-norm estimate $1\le\widehat s\le\kappa$. For any algorithm that solves
all promised Hermitian inputs with constant success probability, there is an
input $I_*$, with $\norm{A^{-1}b}=\Theta(\widehat s)$, on which the algorithm needs
simultaneously
$\Omega\!\left(\kappa\log\frac1\varepsilon\right)$
queries to $U_A$, and
 $\Omega\!\left(\kappa/{\widehat s}\right)$
 queries to $U_b$.
\end{theorem}

\Cref{tab:upper-comparison,tab:lower-bounds} compare our upper and lower
bounds with prior results.

\begin{table}[p]
\centering
\scriptsize
\renewcommand{\arraystretch}{1.10}
\setlength{\tabcolsep}{3.5pt}

\begin{tabular}{@{}
>{\raggedright\arraybackslash}p{0.12\linewidth}
>{\raggedright\arraybackslash}p{0.31\linewidth}
>{\raggedright\arraybackslash}p{0.22\linewidth}
>{\raggedright\arraybackslash}p{0.29\linewidth}@{}}
\toprule
Work
& Main technique
& Queries to $U_b$
& Queries to $U_A$
\\
\midrule

\multicolumn{4}{@{}l}{\textbf{Standard worst-case bounds}}\\
\addlinespace[2pt]

\cite{HHL09}
& Phase estimation + Hamiltonian simulation
& $O\!\left(1/\sqrt{p_{\rm succ}}\right)$
& $O\!\left(\frac{\kappa}{\varepsilon^2\sqrt{p_{\rm succ}}}\right)$ \\

\cite{Ambainis12}
& Variable-time amplitude amplification
& $O\!\left(\kappa\,\operatorname{polylog}
       (\kappa/\varepsilon)\right)$
& $O\!\left(
       \frac{\kappa}{\varepsilon^3}
       \log^3\!\frac{\kappa}{\varepsilon}
       \log^2\!\frac1\varepsilon\right)$ \\

\cite{CKS17}
& Linear combination of quantum walks
& $O\!\left(\frac{1}{\sqrt{p_{\rm succ}}}
       \log\frac{\kappa}{\varepsilon}\right)$
& $O\!\left(
       \frac{\kappa}{\sqrt{p_{\rm succ}}}\,
       \operatorname{polylog}
       (\kappa/\varepsilon)\right)$ \\

\cite{CKS17}
& Gapped phase estimation + quantum walk + VTAA
& $O\!\left(\kappa\,\operatorname{polylog}
       (\kappa/\varepsilon)\right)$
& $O\!\left(\kappa\,\operatorname{polylog}
       (\kappa/\varepsilon)\right)$ \\

\cite{SubasiSommaOrsucci19}
& Randomized adiabatic evolution
& $O\!\left(\frac{\kappa}{\varepsilon}\log\kappa\right)$
& $O\!\left(\frac{\kappa}{\varepsilon}\log\kappa\right)$ \\

\cite{CGJ19}
& Refined VTAA
& $O\!\left(\frac{\log\kappa}{\sqrt{p_{\rm succ}}}\right)$
& $O\!\left(
       \kappa\log\kappa
       \log^2\!\frac{\kappa}{\varepsilon}\right)$ \\

\cite{AnLin22}
& Continuous-time adiabatic evolution
& $O\!\left(\kappa\,\operatorname{polylog}
       (\kappa/\varepsilon)\right)$
& $O\!\left(\kappa\,\operatorname{polylog}
       (\kappa/\varepsilon)\right)$ \\

\cite{LinTong20}
& Eigenstate filtering
& $O\!\left(\kappa\log\frac{\kappa}{\varepsilon}\right)$
& $O\!\left(\kappa\log\frac{\kappa}{\varepsilon}\right)$ \\

\cite{Costa22}
& Discrete adiabatic evolution
& $O\!\left(\kappa\log\frac1\varepsilon\right)$
& {\boldmath$O\!\left(\kappa\log\frac1\varepsilon\right)$} \\

\cite{ChakrabortyMoroliaPeduri23}
& QSVT-based GPE + VTAA
& $O\!\left(\frac{\log\kappa}{\sqrt{p_{\rm succ}}}\right)$
& $O\!\left(
       \kappa\log\kappa
       \log\frac{\kappa}{\varepsilon}\right)$ \\

\cite{Dalzell24}
& Kernel reflection
& $O\!\left(\kappa\log\frac1\varepsilon\right)$
& {\boldmath$O\!\left(\kappa\log\frac1\varepsilon\right)$} \\

\cite{CunninghamRoland24}
& Poisson-distributed phase randomization
& $O\!\left(\kappa\log\frac1\varepsilon\right)$
& {\boldmath$O\!\left(\kappa\log\frac1\varepsilon\right)$} \\

\cite{LowSu26}
& Tunable VTAA + discretized inverse state
& {\boldmath$O\!\left(1/\sqrt{p_{\rm succ}}\right)$}
& $\begin{aligned}
   O\!\biggl(&\kappa\log\frac{2}{\sqrt{p_{\rm succ}}}\\[-2pt]
   &{}\times\left(
       \log\log\frac{4}{\sqrt{p_{\rm succ}}}
       +\log\frac1\varepsilon
       \right)\biggr)
   \end{aligned}$ \\

\cite{LowSu26}
& Block preconditioning
& $O\!\left(\kappa\log\frac1\varepsilon\right)$
& {\boldmath$O\!\left(\kappa\log\frac1\varepsilon\right)$} \\

\cite{Jennings25}
& Randomized adiabatic evolution + filtering
& $O\!\left(\kappa\log\frac1\varepsilon\right)$
& $O\!\left(\kappa\log\frac1\varepsilon\right)$ \\

\cite{YangYuZhang26}
& LC-Schr\"odingerization + block preconditioning
& $O\!\left(\kappa_A\log\frac1\varepsilon\right)$
& $O\!\left(\kappa_A\log\frac1\varepsilon\right)$ \\

\cite{ZhangYangYu26}
& Schr\"odingerization + block preconditioning
& $O\!\left(\kappa_A\log\frac1\varepsilon\right)$
& $O\!\left(\kappa_A\log\frac1\varepsilon\right)$ \\

\cite{Shang26}
& Dissipative Lindbladian
&
$O\!\left(
\kappa^2\log\frac1\varepsilon\,
\frac{\log(\kappa/\varepsilon)}
     {\log\log(\kappa/\varepsilon)}
\right)$
&
$O\!\left(
\kappa^2\log\frac1\varepsilon\,
\frac{\log(\kappa/\varepsilon)}
     {\log\log(\kappa/\varepsilon)}
\right)$
\\

\textbf{This work}
& \textbf{Transducer-based preparation + matrix-only refinement}
& {\boldmath$O\!\left(1/\sqrt{p_{\rm succ}}\right)$}
& {\boldmath$O\!\left(\kappa\log\frac1\varepsilon\right)$} \\

\midrule
\multicolumn{4}{@{}l}{
\textbf{Instance-dependent / beyond-$\kappa$ bounds}
\quad (not directly comparable to the worst-case bounds above)}
\\
\addlinespace[2pt]

\cite{Li25}
& Filtering with an effective gap
& $O(g/\varepsilon^2)$
& $O(g/\varepsilon^2)$ \\

\cite{DalzellLiSu26}
& Effective truncation
& $O\!\left(\kappa_{\rm eff}/\|A^{-1}b\|\right)$
& $\begin{aligned}
   O\!\biggl(&\kappa_{\rm eff}
   \log\frac{\kappa_{\rm eff}}{\|A^{-1}b\|}\\[-2pt]
   &{}\times\left[
   \log\log\frac{\kappa_{\rm eff}}{\|A^{-1}b\|}
   +\log\frac1\varepsilon\right]\biggr)
   \end{aligned}$ \\

\cite{DalzellLiSu26}
& Filtering with an effective gap
& $O\!\left(
       \frac{g}{\varepsilon}\log\frac1\varepsilon
       \right)$
& $O\!\left(
       \frac{g}{\varepsilon}\log\frac1\varepsilon
       \right)$ \\

\bottomrule
\end{tabular}

\caption{
Upper bounds for quantum linear-system algorithms.
Here $\|A\|\le 1$,
$p_{\rm succ}=\|A^{-1}b\|^2/\kappa^2$,
$\ket{x}=A^{-1}b/\|A^{-1}b\|$, and
$g=\|A^{-1\dagger}\ket{x}\|$.
For the Schr\"odingerization results
\cite{YangYuZhang26,ZhangYangYu26},
$\kappa_A:=\alpha_A\alpha_{A^{-1}}$ denotes the supplied upper bound on
the spectral condition number, with $\alpha_A\ge\|A\|$ and
$\alpha_{A^{-1}}\ge\|A^{-1}\|$.  When these bounds are tight,
$\kappa_A=\kappa$.
The lower part uses instance-dependent parameters and is not directly
comparable with the upper part.
}
\label{tab:upper-comparison}
\end{table}
\begin{table*}[t]
\centering
\small
\renewcommand{\arraystretch}{1.18}
\setlength{\tabcolsep}{4pt}

\begin{tabular}{@{}p{0.17\textwidth}
                    p{0.25\textwidth}
                    p{0.38\textwidth}
                    p{0.13\textwidth}@{}}
\toprule
Work
& Queries to $U_b$
& Queries to $U_A$
& Same input?
\\
\midrule

\cite{SommaSubasi21}
& $\Omega(\kappa/s)$ for state verification
& ---
& --- \\

\cite{OrsucciDunjko21}
& $\Omega(\min\{\kappa,N\})$
& $\Omega(\min\{\kappa,N\})$ for $U_A$ and sparse access
& Not shown \\

\cite{HarrowKothari}
& ---
& $\Omega\!\left(
    \kappa\log\frac1\varepsilon
  \right)$
& --- \\

\cite{LowSu26}
& $\Omega(\kappa/s)$
& ---
& --- \\

\cite{Mori26}
& ---
& $\Omega(\kappa\sqrt{d_{\rm sp}})$
  for constant $\varepsilon$ and sparse access
& --- \\

\cite{Patel26}
& ---
& $\Omega\!\left(
    \kappa\sqrt{d_{\rm sp}}
    \log\frac1\varepsilon
  \right)$ for sparse access
& --- \\

\textbf{This work}
& {\boldmath$\Omega(\kappa/s)$}
& {\boldmath$\Omega\!\left(
    \kappa\log\frac1\varepsilon
  \right)$}
& \textbf{Yes} \\

\bottomrule
\end{tabular}

\caption{
Selected lower bounds relevant to QLS.
Here $\|A\|\le1$, $s=\|A^{-1}b\|$, $N$ is the dimension, and
$d_{\rm sp}$ is the matrix sparsity.
Only our result proves both bounds on one input.
}
\label{tab:lower-bounds}
\end{table*}

\subsection{Technique overview}

The main difficulty is to make the dependence on the target precision
$\varepsilon$ appear only in queries to the matrix oracle.
A standard high-precision QLS procedure repeatedly uses both $U_A$ and $U_b$,
which is incompatible with the optimal $O(\kappa/s)$ bound for
state-preparation queries.

Our algorithm separates the computation into two stages.
First, using both oracles, we prepare only a \emph{constant-accuracy} state.
To control the two oracle costs separately, we formulate this preparation as a
canonical transducer.
Its finite implementation lets us charge matrix-dependent and
state-preparation-dependent query subspaces independently, yielding
$O(\kappa)$ queries to $U_A$ and $O(\kappa/s)$ queries to $U_b$.

The key point is that the resulting state has substantially more structure
than its constant-error guarantee suggests.
After projection onto a suitable auxiliary kernel, its direction is
\emph{exactly} the desired solution direction.
This exactness follows from a real Krylov-space invariant and is what makes
high-precision refinement possible without querying $U_b$ again.

In the second stage, a matrix-only polynomial filter suppresses the component
orthogonal to the kernel while preserving the desired kernel component.
Reducing the remaining error to $\varepsilon$ costs
$O(\kappa\log(1/\varepsilon))$ queries to $U_A$, but no additional queries to
$U_b$.
Thus the precision overhead is completely decoupled from state preparation.

For the lower bound, it is not enough to combine separate hard instances for
the two oracles.
We start from a family that is hard for matrix queries and prove that every
instance in the family is also hard for state-preparation queries.
A small perturbation of the right-hand side changes the solution by a constant
amount while changing the preparation oracle by only $O(s/\kappa)$.
A one-sided hybrid argument therefore gives an $\Omega(\kappa/s)$ lower bound
for the original instance itself.
Choosing a matrix-hard member of the same family yields a single instance that
simultaneously forces both lower bounds.

\subsection{Related work}\label{sec:related}

\paragraph{Recent developments in QLS.}
Recent work has explored several directions beyond the standard
worst-case QLS algorithms.  Low and Su~\cite{LowSu26} obtained the
optimal query complexity in the preparation of the right-hand side,
while retaining nearly optimal matrix-query complexity.  Dalzell, Li,
and Su~\cite{DalzellLiSu26} developed instance-dependent solvers whose
complexity can go beyond the worst-case condition number.  Other recent
approaches include a purely dissipative formulation of QLS
\cite{Shang26} and Schr\"odingerization-based solvers
\cite{YangYuZhang26,ZhangYangYu26}, the latter attaining optimal
matrix-query complexity without variable-time amplitude amplification.
Recent work has also studied constant factors and practical resource
estimates: Costa et al.~\cite{Costa25} numerically analyzed the
discrete-adiabatic solver of \cite{Costa22}.
These works address complementary aspects of QLS; our focus is on
simultaneously attaining the optimal query complexities with respect to
the two distinct input oracles.

\paragraph{Transducers and query-optimal quantum algorithms.}
Transducers were introduced by Belovs, Jeffery, and
Yolcu~\cite{BJY24} as an idealized model for designing
quantum algorithms with favorable composition and error properties.
Apers, Roland, and Zhang~\cite{ApersRolandZhang26} recently developed
zero-error transducers for electric-flow sampling and, more generally,
for reflections about intersections of subspaces, leading to improved
quantum-walk algorithms with optimal error dependence.
Transducers have also been used to obtain optimal query complexities
for time-dependent Hamiltonian simulation
\cite{ChenGaoWangZhou26}, guided ground-state energy estimation
\cite{JefferyWitteveen26}, and Lindbladian simulation
\cite{WangYe26,ChenGaoWangZhouLindblad26,Kitchen26}.
Related work has further studied transducer-based error reduction and
gate-efficient implementations
\cite{BelovsJeffery26,ChenGaoJiLiWangZhou26}.
Our work applies this framework to QLS, where the main challenge is to
optimize the costs of two distinct input oracles simultaneously.

\section{Problem and preliminaries}\label{sec:model}\label{sec:tools}

\subsection{Notations}

We use $\norm{\cdot}$ for the Euclidean norm of a vector and the operator norm of a matrix. For a column vector $v=(v_0,\ldots,v_{d-1})^{\intercal}$, write
\[
 \ket v_{\rD}=\sum_{j=0}^{d-1}v_j\ket j_{\rD}
\]
on a $\lceil\log_2d\rceil$-qubit data register $\rD$, with zero amplitudes on unused basis states. This notation does not normalize $v$: it represents a quantum state when $\norm v=1$. We use column vectors in matrix identities and kets for register states. The inner product is conjugate-linear in its first argument. All logarithms are natural unless a base is displayed. We write $M^+$ for the Moore--Penrose pseudoinverse and $\Pi_{\ker M}$ for the orthogonal projector onto the kernel of a Hermitian matrix \cite{Higham02,Watrous18}.

\begin{definition}[Block-encoding, {\cite[Definition 43]{GSLW19}}]\label[definition]{def:input}
A unitary $U_A$ is an $(\alpha,a,\delta_A)$-\emph{block-encoding} of $A\in\C^{d\times d}$ if
\begin{equation}\label{eq:block}
 \left\|A-\alpha\left(\bra{0^a}\otimes I\right)U_A\left(\ket{0^a}\otimes I\right)\right\|\le\delta_A,
\end{equation}
where $\alpha>0$, $a$ is a nonnegative integer, and $\delta_A\ge0$. The encoding is exact when $\delta_A=0$.
\end{definition}

\subsection{Problem definition}
\begin{problem}[Quantum linear systems with a solution-norm estimate]\label[problem]{prob:qlsa}
Let $A\in\C^{d\times d}$ be invertible and $b\in\C^d$ be a unit vector.
Given query access to an exact
$(\alpha,a,0)$-block-encoding $U_A$ of $A$ and a state-preparation unitary
$U_b$ preparing $\ket b$, together with $\kappa\ge2$, an estimate
$\widehat s$, and $0<\varepsilon<1/2$ satisfying
$\alpha\norm{A^{-1}}\le\kappa$ and
$\alpha\norm{A^{-1}b}/2\le\widehat s\le2\alpha\norm{A^{-1}b}$, 
the task is to 
output a state $\ket{\widetilde x}$ such
that
\[
 \left\|\widetilde x-
 \frac{A^{-1}b}{\norm{A^{-1}b}}\right\|\le\varepsilon
\]
with probability at least $2/3$.
\end{problem}

\begin{proposition}[Normalization and Hermitian reduction]\label[proposition]{prop:hermitian-reduction}\label{sec:dilation}
Let $(A,b,U_A,U_b,\alpha,a,\kappa,\widehat s,\varepsilon)$ be as in
\cref{prob:qlsa}, except that the estimate may satisfy the relaxed bounds
\[
 \frac{3\alpha\norm{A^{-1}b}}8\le\widehat s\le
 \frac{5\alpha\norm{A^{-1}b}}2,
\]
and define
\begin{equation}\label{eq:hermitian-dilation}
 \widehat A=\begin{pmatrix}0&A/\alpha\\ A^\dagger/\alpha&0\end{pmatrix},
 \qquad
 \widehat b=\binom b0,
\qquad
 U_{\widehat A}=\begin{pmatrix}0&U_A\\ U_A^\dagger&0\end{pmatrix},
\qquad
 U_{\widehat b}=I\otimes U_b.
\end{equation}
Then:
\begin{itemize}
\item $(\widehat A,\widehat b,U_{\widehat A},U_{\widehat b},1,a,\kappa,
\widehat s,\varepsilon/2)$ satisfies the same relaxed input conditions. In
addition, $\widehat A$ is invertible and Hermitian,
$U_{\widehat A}$ is an exact $(1,a,0)$-block-encoding of $\widehat A$,
$U_{\widehat b}$ prepares $\ket{\widehat b}$, and
\[
 \left\|\widehat A\right\|\le1,
 \qquad
 \left\|\widehat A^{-1}\right\|\le\kappa,
 \qquad
 \frac{3\left\|\widehat A^{-1}\widehat b\right\|}{8}
 \le\widehat s\le
 \frac{5\left\|\widehat A^{-1}\widehat b\right\|}{2}.
\]
\item Suppose that, on these parameters, one run of an algorithm succeeds with
probability at least $2/3$ and, on success, outputs $\ket{\widetilde y}$ satisfying
\[
 \left\|\widetilde y-
 \frac{\widehat A^{-1}\widehat b}
 {\left\|\widehat A^{-1}\widehat b\right\|}\right\|
 \le\frac{\varepsilon}{2},
\]
then at most three independent runs of the algorithm, each followed by
measurement of the first system qubit, output
$\ket{\widetilde x}$ with probability greater than $2/3$ such that
\[
 \left\|\widetilde x-
 \frac{A^{-1}b}{\left\|A^{-1}b\right\|}\right\|\le\varepsilon.
\]
\item If one run of the algorithm uses $q_A$ queries to $U_{\widehat A}$,
$q_b$ queries to $U_{\widehat b}$, and $g$ extra gates, the resulting costs
are $O(q_A)$ queries to $U_A$, $O(q_b)$ queries to $U_b$, and
$O(g+q_A+q_b)$ extra gates.
\end{itemize}
\end{proposition}
\begin{proof}
The exact block-encoding gives $\norm A\le\alpha$, and direct calculation gives
\[
 \widehat A^2
 =\alpha^{-2}\diag\left(AA^\dagger,A^\dagger A\right),
 \qquad
 \widehat A^{-1}
 =\alpha\begin{pmatrix}0&\left(A^\dagger\right)^{-1}\\ A^{-1}&0\end{pmatrix}.
\]
Thus $\widehat A=\widehat A^\dagger$,
$\norm{\widehat A}=\norm A/\alpha\le1$, and
$\norm{\widehat A^{-1}}=\alpha\norm{A^{-1}}\le\kappa$.  Also,
\[
 \widehat A^{-1}\widehat b=\binom0{\alpha A^{-1}b},
\qquad
 \left\|\widehat A^{-1}\widehat b\right\|
 =\alpha\left\|A^{-1}b\right\|,
 \qquad
 \frac{\widehat A^{-1}\widehat b}
      {\left\|\widehat A^{-1}\widehat b\right\|}
 =\binom0{A^{-1}b/\left\|A^{-1}b\right\|}.
\]
Compressing the $a$ signal qubits of $U_{\widehat A}$ gives $\widehat A$,
while $U_{\widehat b}\ket0\ket{0^n}=\ket0\ket b$.  This proves the first item.

For a successful output, let $y_1=(\bra1\otimes I)\widetilde y$.  Then
\[
 \left\|y_1-\frac{A^{-1}b}{\left\|A^{-1}b\right\|}\right\|
 \le\frac{\varepsilon}{2},
 \qquad
 \left\|y_1\right\|\ge1-\frac{\varepsilon}{2}\ge\frac34.
\]
The first system qubit is therefore $1$ with probability at least $9/16$,
and its normalized remaining state satisfies
\begin{align*}
 \left\|\frac{y_1}{\left\|y_1\right\|}
       -\frac{A^{-1}b}{\left\|A^{-1}b\right\|}\right\|
 &\le \left|1-\left\|y_1\right\|\right|
      +\left\|y_1-\frac{A^{-1}b}{\left\|A^{-1}b\right\|}\right\|\\
 &\le2\left\|y_1-\frac{A^{-1}b}{\left\|A^{-1}b\right\|}\right\|
 \le\varepsilon.
\end{align*}
Each run succeeds with probability at least $(2/3)(9/16)=3/8$, so three
independent runs succeed with probability at least $1-(5/8)^3>2/3$.
Each call to $U_{\widehat A}$ uses $O(1)$ calls to $U_A$ or $U_A^\dagger$ and
$O(1)$ extra gates, and each call to $U_{\widehat b}$ uses one call to $U_b$.
The third item follows.
\end{proof}

By \cref{prop:hermitian-reduction}, throughout the upper-bound construction
we assume that $U_A$ is an exact $(1,a,0)$-block-encoding of a Hermitian
matrix $A$ with $\norm A\le1$ and $\norm{A^{-1}}\le\kappa$. We write
$s=\norm{A^{-1}b}$, so $1\le s\le\kappa$.

\subsection{Complexity and error conventions}

We count resources as follows.
\begin{itemize}
\item A call to $U_A$ or $U_A^\dagger$ or their controlled version is one \emph{matrix query}; a call to $U_b$, $U_b^\dagger$, or their controlled version is one \emph{vector query}.
\item An \emph{extra gate} is a one- or two-qubit gate outside the input-oracle calls, which allows exact rotations. 
\end{itemize}

For lower bounds we allow mixed conditional outputs and the weaker guarantee
\begin{equation}\label{eq:trace-guarantee}
 \Dtr\left(\rho_{\rm out},\proj x\right)\le\varepsilon,\qquad \text{ where }
 \Dtr\left(\rho,\sigma\right)=\tfrac12\norm{\rho-\sigma}_1.
\end{equation}
Success probability remains at least $2/3$.  We assume that success is
heralded by a classical output flag; conditioned on that flag, the output
state satisfies \eqref{eq:trace-guarantee}.  This is the success flag used in
the lower-bound reductions below.  The vector guarantee implies this
trace-distance guarantee, as $\Dtr(\proj u,\proj v)\le\norm{u-v}$ for unit
vectors \cite[Chapter 3]{Watrous18}.

\subsection{Polynomial matrix transformations}

\begin{lemma}[Eigenvalue transformation, {\cite[Corollary 18]{GSLW19}}]\label[lemma]{lem:qsvt}
Let $U_A$ be an $(1,a,0)$-block-encoding of a Hermitian matrix $A$. For a real even polynomial $p$ of degree at most $d$ with $|p(x)|\le1$ on $[-1,1]$, an exact block-encoding of $p(A)$ can be implemented using $O(d+1)$ queries to $U_A$ and $O((d+1)(a+1))$ extra gates.
\end{lemma}

\begin{lemma}[Polynomial approximations of negative power functions, {\cite[Corollary 67]{GSLW19}}]\label[lemma]{lem:inverse-square}
For $0<\delta\le1/2$ and $0<\varepsilon<1/2$, there is a real even
polynomial $p_{\delta,\varepsilon}$ such that
\begin{itemize}
\item $p_{\delta,\varepsilon}$ has degree
$O(\delta^{-1}\log(1/\varepsilon))$;
\item $|p_{\delta,\varepsilon}(x)|\le1$ for $|x|\le1$;
\item $|p_{\delta,\varepsilon}(x)-\delta^2/(2x^2)|\le\varepsilon$ for
$\delta\le|x|\le1$.
\end{itemize}
\end{lemma}

\subsection{Transducers}\label{sec:finite-tool}

\begin{definition}[Transducer and catalyst, {\cite[Section 5.1]{BJY24}}]\label[definition]{def:transducer}
Let $\cH$ and $\mathcal L$ be finite-dimensional Hilbert spaces, called the
\emph{public space} and \emph{private space}, respectively. A \emph{transducer}
is a unitary $S$ on $\cH\oplus\mathcal L$. For $\xi,\zeta\in\cH$ and
$\omega\in\mathcal L$, if
\begin{equation}\label{eq:transducer-definition}
 S\left(\xi\oplus\omega\right)=\zeta\oplus\omega.
\end{equation}
then $\omega$ is a \emph{catalyst} for $\xi$, and we write
$\xi\transduce{S}\zeta$ (or $S\colon\xi\transduce{}\zeta$). The induced map
on $\cH$ is the \emph{transduction action} of $S$, denoted by
$S\DownTransduce_{\cH}$.
\end{definition}

By \cite[Theorem 5.1]{BJY24}, in finite dimensions
$S\DownTransduce_{\cH}$ is a well-defined unitary on
$\cH$, and for every $\xi\in\cH$ there exists $\omega_\xi\in\mathcal L$
such that
\[
 S(\xi\oplus\omega_\xi)
 =\bigl(S\DownTransduce_{\cH}\xi\bigr)\oplus\omega_\xi
\]
The vector $\omega_\xi$ need not be normalized.

\begin{definition}[Canonical transducer, {\cite[Sections 7.1--7.2]{BJY24}}]\label[definition]{def:canonical}
Let
$\mathcal L=\mathcal L_0\oplus\mathcal L_1\oplus\mathcal L_2$, where
$\mathcal L_0$ is the non-query subspace, and let $O_i$ be a unitary input
oracle on $\mathcal L_i$ for $i=1,2$. A \emph{canonical transducer} is
\begin{equation}\label{eq:canonical-definition}
 S=S^\circ\left(I_{\cH}\oplus I_0\oplus O_1\oplus O_2\right),
\end{equation}
where $I_0$ is the identity on $\mathcal L_0$. The \emph{work unitary}
$S^\circ$ is independent of the input quantum oracles.
\end{definition}

The work unitary $S^\circ$ is an explicit circuit rather than an input
oracle.  The finite compiler uses a controlled circuit for $S^\circ$, whose
gates are counted as extra gates; only the calls to $O_1$ and $O_2$ are
counted as input queries.

Each $O_i$ is controlled on its query label; multiplicity spaces share the same
oracle call. For a chosen catalyst
$\omega=\omega_0\oplus\omega_1\oplus\omega_2$, define
\begin{equation}\label{eq:transducer-complexities}
 W:=\norm\omega^2,
 \qquad L_i:=\norm{\omega_i}^2\quad(i=1,2).
\end{equation}
Following \cite[Section 7.1]{BJY24}, $W$ is the \emph{transduction complexity},
and $L_i$ is the $i$th \emph{partial Las Vegas query complexity}.

\begin{lemma}[Fixed-accuracy implementation, {\cite[Theorem 7.1]{BJY24}}]\label[lemma]{thm:bjy}
Let $S$ be a canonical transducer with input oracles $O_1$ and $O_2$, represented
using a constant-size label register and a shared workspace. For each unit
input $\xi$ under consideration, choose a catalyst
$\omega(\xi)=\omega_0(\xi)\oplus\omega_1(\xi)\oplus\omega_2(\xi)$. Suppose
known $W,L_1,L_2$ satisfy
\[
 \norm{\omega(\xi)}^2\le W,
 \qquad \norm{\omega_i(\xi)}^2\le L_i,
 \qquad 1\le L_i\le W
 \quad (i=1,2).
\]
Suppose also that a controlled circuit for $S^\circ$ uses at most $g\ge1$
extra gates.
For every fixed $0<\varepsilon_0<1$, there is a unitary $U_{\rm fin}$ such that
\begin{equation}\label{eq:fixed-transducer-error}
 \norm{U_{\rm fin}\left(\xi\otimes\ket0_{\rm aux}\right)
       -\left(S\DownTransduce_{\cH}\xi\right)\otimes\ket0_{\rm aux}}
 \le\varepsilon_0.
\end{equation}
It uses
\begin{itemize}
\item $O(L_i)$ controlled queries to each $O_i$;
\item $O(W(g+1))$ extra gates and $O(\log(2+W))$ auxiliary qubits.
\end{itemize}
The same $U_{\rm fin}$ works for all such $\xi$: it depends on the bounds but not
on $\omega(\xi)$. Its auxiliary label and control operations can be chosen real.
\end{lemma}

\begin{proof}
Choose the compiler and oracle budgets separately in
\cite[Theorem 7.1]{BJY24}.  For fixed $\varepsilon_0$, its error bound allows
$O(L_i)$ queries to $O_i$.  Its $O(W)$ uses of the controlled circuit for
$S^\circ$ contribute $O(Wg)$ gates, and the remaining compiler overhead is
$O(W)$.  All compiler gates other than the controlled $S^\circ$ circuits and
oracle queries can be chosen real by
\cref{lem:real-compiler}.
\end{proof}


\section{Algorithm overview}\label{sec:algorithm}

We give the solver for these inputs, then construct its preparation and refinement circuits in the next two sections. Only the preparation uses the state-preparation oracle.

\subsection{Notations for the Algorithm}

We use a data register $\rD$ to store the input vector as
$\ket b_{\rD}=\sum_{j=0}^{d-1}b_j\ket j_{\rD}$. A two-qubit register $\rG$
stores a component label: $\ket g_{\rG}$ denotes component $g$ for
$g\in\{0,1,2\}$. The remaining basis state $\ket3_{\rG}$ is unused. On $\rG\rD$, define
\begin{equation}\label{eq:graph}
\begin{aligned}
 H_{\rG\rD}
 &=\left(\ket0\bra1+\ket1\bra0\right)_{\rG}\otimes A_{\rD}
   -\kappa^{-1}\left(\ket0\bra2+\ket2\bra0\right)_{\rG}
      \otimes I_{\rD},\\
 \ket e_{\rG\rD}&=\ket1_{\rG}\ket b_{\rD},\qquad
 P=\Pi_{\ker H},\qquad
 \ket u_{\rG\rD}=\frac{P\ket e_{\rG\rD}}{\norm{P\ket e_{\rG\rD}}}.
\end{aligned}
\end{equation}

Besides $\rG\rD$, the algorithm uses disjoint ancillary registers $\rA$, $\rF$,
and $\rC$ for preparation, filtering, and correction, respectively. For any
register $\mathsf{R}$, $\ket0_{\mathsf{R}}$ denotes its all-zero state. A
circuit subscript lists the registers on which it acts.

The acceptance projector is
\begin{equation}\label{eq:success-projector}
 \Pi_{\rm succ}
 =\proj0_{\rA}\otimes\proj0_{\rF}\otimes\proj0_{\rC}
       \otimes\proj2_{\rG}\otimes I_{\rD}.
\end{equation}
The solver uses the kernel of $H$ to connect the input $\ket b_{\rD}$ to the
solution: it prepares a component along $\ket u_{\rG\rD}$, filters out the
orthogonal component, and converts the $\rG=2$ component into $\ket x_{\rD}$.

\subsection{Solver and subroutines}

\paragraph{Description of the algorithm.}
A \emph{run} starts with all five registers set to zero, applies the three
subroutines below, and ends with the measurement of
$\{\Pi_{\rm succ},I-\Pi_{\rm succ}\}$. Ideally, the subroutines perform the
following tasks.
\begin{itemize}
\item \emph{Preparation.}
$\Coarse$ prepares $\ket\Psi_{\rA\rG\rD}$ such that
\[
 \begin{aligned}
 \left(\bra0_{\rA}\otimes I_{\rG\rD}\right)
 \ket\Psi_{\rA\rG\rD}
 &=\beta\ket u_{\rG\rD}+\ket h_{\rG\rD},\\
 P\ket h_{\rG\rD}&=0,\qquad \beta\in[1/32,1].
 \end{aligned}
\]
Thus $\beta\ket u$ is the desired kernel component, while $\ket h$ is removed
by the next step.

\item \emph{Kernel filtering.}
$\FilterKernel(U_A;\kappa,\eta)_{\rF\rG\rD}$ is an
$\eta$-approximate block-encoding of $P$, with signal register $\rF$. It maps
$\beta\ket u+\ket h$ to $\beta\ket u+O(\eta)$ in its $\rF=0$ block.

\item \emph{Solution correction.}
$\CorrectSolution(U_A;\kappa,\eta)_{\rC\rD}$ is an
$\eta/2$-approximate block-encoding of
$C_\star=(I+\kappa^{-2}A^{-2})/4$, with signal register $\rC$.
By \cref{lem:geometry}, the $\rG=2$ component of $\ket u_{\rG\rD}$ is
proportional to $A(A^2+\kappa^{-2}I)^{-1}b$, and
\[
 C_\star\ket{A(A^2+\kappa^{-2}I)^{-1}b}_{\rD}
 =\frac14\ket{A^{-1}b}_{\rD}=\frac{s}{4}\ket x_{\rD}.
\]
\end{itemize}

Finally, $\Pi_{\rm succ}$ selects $\rA=\rF=\rC=0$ and $\rG=2$, leaving
$\rD$ as the output register. The preparation uses $O(\kappa)$ matrix queries
and $O(\kappa/s)$ vector queries; the filter and correction each
use $O(\kappa\log(1/\eta))$ matrix queries. \Cref{prop:refinement} proves the
error and success probability, and $72000$ independent runs reduce the total
failure probability below $1/3$.

\begin{algorithm}[H]
\caption{$\Solver(U_A,U_b;\kappa,\widehat s,\varepsilon)$}\label{alg:solver}
\begin{algorithmic}[1]
\Require A normalized Hermitian input with
         $3s/8\le\widehat s\le5s/2$, where $s=\norm{A^{-1}b}$.
\Ensure On success, an $\varepsilon$-approximation to $\ket x$ in $\rD$.
\State $\eta\gets\varepsilon/1024$.
\For{$j=1,\ldots,72000$}
  \State Initialize $\rA,\rF,\rC,\rG,\rD$ to zero.
  \State Apply $\Coarse(U_A,U_b;\kappa,\widehat s)_{\rA\rG\rD}$.
  \State Apply $\FilterKernel(U_A;\kappa,\eta)_{\rF\rG\rD}$.
  \State Apply $\CorrectSolution(U_A;\kappa,\eta)_{\rC\rD}$.
  \State Measure $\{\Pi_{\rm succ},I-\Pi_{\rm succ}\}$.
  \If{the outcome is $\Pi_{\rm succ}$} \State \Return $\rD$ and success. \EndIf
\EndFor
\State \Return failure.
\end{algorithmic}
\end{algorithm}



\section{Preparing the kernel component}\label{sec:graph}\label{sec:transducer}

This section presents \cref{alg:coarse} and proves its correctness.

\begin{algorithm}[H]
\caption{$\Coarse(U_A,U_b;\kappa,\widehat s)$}\label{alg:coarse}
\begin{algorithmic}[1]
\Require A normalized Hermitian input and an estimate
         $\widehat s$ satisfying $3s/8\le\widehat s\le5s/2$, where
         $s=\norm{A^{-1}b}$.
\Ensure A state $\ket\Psi_{\rA\rG\rD}$ satisfying
        \eqref{eq:coarse-guarantee}.
\State Initialize $\rA,\rG,\rD$ to zero.
\State Construct the preparation transducer $S$ from
       \cref{prop:preparation-transducer}.
\State $U_{\rm fin}\gets\Finite(S;\kappa,\widehat s)$
       \Comment{See \cref{prop:finite-preparation}}
\State Prepare $\ket1_{\rG}\ket b_{\rD}$, apply $U_{\rm fin}$,
       and \Return $\ket\Psi_{\rA\rG\rD}$.
\end{algorithmic}
\end{algorithm}

The action needed from $S$ is
\[
 \psi=\left(S\DownTransduce_{\rG\rD}\right)e,
 \qquad \norm\psi=1,\qquad \langle u,\psi\rangle>\frac1{30}.
\]
The routine $\Finite$ implements this action by interleaving controlled
applications of the work part of $S$ with calls to its two oracle parts. For
$U_{\rm fin}=\Finite(S;\kappa,\widehat s)$, it guarantees
\[
 \left\|U_{\rm fin}\left(\ket0_{\rA}\ket e_{\rG\rD}\right)
       -\ket0_{\rA}\ket\psi_{\rG\rD}\right\|<10^{-3}.
\]
For fixed classical parameters $\kappa$ and $\widehat s$, the circuit uses
$O(\kappa)$ matrix queries and $O(\kappa/s)$ vector queries; see
\cref{prop:preparation-transducer,prop:finite-preparation}.

The goal of this section is to prove the following proposition.

\begin{proposition}[Kernel-component preparation]\label[proposition]{prop:coarse}
Let $s=\norm{A^{-1}b}$, and suppose the estimate $\widehat s$
satisfies $3s/8\le\widehat s\le5s/2$.
For $P$ and $\ket u$ defined in \cref{eq:graph}, \cref{alg:coarse}
prepares a unit vector $\ket\Psi_{\rA\rG\rD}$ satisfying
\begin{equation}\label{eq:coarse-guarantee}
\begin{aligned}
 \left(\bra0_{\rA}\otimes I_{\rG\rD}\right)
 \ket\Psi_{\rA\rG\rD}
 &=\beta\ket u_{\rG\rD}+\ket h_{\rG\rD},\\
 P\ket h_{\rG\rD}&=0,\qquad \beta\in[1/32,1].
\end{aligned}
\end{equation}
It uses $O(\kappa)$ matrix queries, $O(\kappa/s)$ vector queries, and
\begin{equation*}
       O\!\left(\kappa\left(a+1\right)+\frac\kappa s\left(n+1\right)\right)
\end{equation*}
extra one- and two-qubit gates.
\end{proposition}

\subsection{Auxiliary matrix and reflections}

\begin{lemma}[Kernel projection and spectral gap]\label[lemma]{lem:geometry}
For $H,e,P,u$ in \eqref{eq:graph}, the following properties hold.
\begin{itemize}
\item \emph{Projection and pseudoinverse.}
\begin{align}
 Pe&=\left(0,\ \kappa^{-2}\left(A^2+\kappa^{-2}I\right)^{-1}b,
             \ \kappa^{-1}A\left(A^2+\kappa^{-2}I\right)^{-1}b\right),\label{eq:graph-identities}\\
 H^+e&=\left(A\left(A^2+\kappa^{-2}I\right)^{-1}b,0,0\right).\notag
\end{align}
\item \emph{Norm bounds.} The projection satisfies
\begin{equation}\label{eq:gamma}
 \frac{s^2}{2\kappa^2}\le\norm{Pe}^2\le\frac{s^2}{\kappa^2},
 \qquad 0<\norm{Pe}^2\le\frac12,
\end{equation}
so $u$ is well defined. Moreover,
\begin{equation}\label{eq:pseudoinverse-w}
 H^+u=0,\qquad \norm{H^+e}\le\norm{A^{-1}b}.
\end{equation}
\item \emph{Spectrum.} Every nonzero eigenvalue $\lambda$ of $H$ satisfies
\[
 \frac{\sqrt2}{\kappa}\le|\lambda|\le\sqrt{1+\kappa^{-2}}.
\]
\end{itemize}
\end{lemma}

\begin{proof}
\emph{Projection and pseudoinverse.}
Let $D=A^2+\kappa^{-2}I$.
Set
\[
 p=\left(0,\kappa^{-2}D^{-1}b,\kappa^{-1}AD^{-1}b\right),
 \qquad
 y=\left(AD^{-1}b,0,0\right).
\]
Since $D$ is a polynomial in $A$, the operators $A$ and $D^{-1}$ commute.
Consequently,
\[
 Hp=\left(\kappa^{-2}AD^{-1}b
              -\kappa^{-2}AD^{-1}b,0,0\right)=0,
\]
and hence $p\in\ker H$. The identity $D=A^2+\kappa^{-2}I$ also gives
$b-\kappa^{-2}D^{-1}b=A^2D^{-1}b$, so
\[
 e-p
 =\left(0,A^2D^{-1}b,-\kappa^{-1}AD^{-1}b\right)
 =Hy\in\operatorname{Ran}H.
\]
Because $H$ is Hermitian, $\operatorname{Ran}H=(\ker H)^\perp$. Thus
$e=p+(e-p)$ is the orthogonal decomposition of $e$ into its kernel and range
components, which proves $Pe=p$ and the first identity in
\eqref{eq:graph-identities}.

To identify the pseudoinverse, write a vector in $\ker H$ as
$\xi=(\xi_0,\xi_1,\xi_2)$. The second component of $H\xi=0$ is
$A\xi_0=0$. Since $A$ is invertible, $\xi_0=0$, and therefore $y$ is
orthogonal to every vector in $\ker H$. Hence $y\in\operatorname{Ran}H$ and,
since $Pe=p$ and $e-p=Hy$, we have $Hy=(I-P)e$.  The spectral
decomposition of the Hermitian matrix $H$ gives
$H^+P=0$ and $H^+H=\Pi_{\operatorname{Ran}H}$. Therefore
\[
\begin{aligned}
 H^+e
 &=H^+Pe+H^+(I-P)e\\
 &=H^+Hy\\
 &=(H^+H)y\\
 &=\Pi_{\operatorname{Ran}H}y
 =y,
\end{aligned}
\]
which proves the second identity in \eqref{eq:graph-identities}.

\emph{Norm bounds.}
Using the commutation of $A$ and $D^{-1}$, the squared projection norm is
\[
\begin{aligned}
 \norm{Pe}^2
 &=\kappa^{-4}\langle b,D^{-2}b\rangle
   +\kappa^{-2}\langle b,A^2D^{-2}b\rangle\\
 &=\kappa^{-2}\langle b,D^{-2}
          \left(\kappa^{-2}I+A^2\right)b\rangle
 =\kappa^{-2}\langle b,D^{-1}b\rangle.
\end{aligned}
\]
The bound $\kappa\ge\norm{A^{-1}}$ implies
$A^2\succeq\kappa^{-2}I$. Therefore
$A^2\preceq D\preceq2A^2$ and $D\succeq2\kappa^{-2}I$. Inverting these
positive-definite operator inequalities reverses their order and gives
\[
 \tfrac12 A^{-2}\preceq D^{-1}\preceq A^{-2},\qquad
 D^{-1}\preceq\tfrac12\kappa^2I.
\]
Taking the expectation in the unit vector $b$, multiplying by
$\kappa^{-2}$, and using $\langle b,A^{-2}b\rangle=s^2$ proves
\[
 \frac{s^2}{2\kappa^2}\le\norm{Pe}^2\le\frac{s^2}{\kappa^2},
 \qquad \norm{Pe}^2\le\frac12.
\]
The lower bound is positive, so $u$ is well defined. Since $u\in\ker H$,
we have $H^+u=0$. For every eigenvalue $\lambda$ of $A$, the operator
$AD^{-1}$ has eigenvalue $\lambda/(\lambda^2+\kappa^{-2})$, whose absolute
value is at most $1/|\lambda|$. Hence
\[
 \norm{H^+e}^2=\norm{AD^{-1}b}^2
 \le\langle b,A^{-2}b\rangle=\norm{A^{-1}b}^2,
\]
which proves \eqref{eq:pseudoinverse-w}.

\emph{Spectrum.}
Because $A$ is Hermitian, its
eigenvectors form an orthonormal basis. For an eigenvector $\phi$ with
$A\phi=\lambda\phi$, the three-dimensional space
\[
 \operatorname{span}\{\ket0\otimes\phi,\ket1\otimes\phi,
 \ket2\otimes\phi\}
\]
is invariant under $H$. In the displayed order, the restriction of $H$ to
this space is the matrix
\[
 H_\lambda=\begin{pmatrix}
  0&\lambda&-\kappa^{-1}\\
  \lambda&0&0\\
  -\kappa^{-1}&0&0
 \end{pmatrix}.
\]
Its characteristic polynomial is
\[
 \det\!\left(zI-H_\lambda\right)
 =z\left(z^2-\lambda^2-\kappa^{-2}\right),
\]
so its eigenvalues are $0$ and $\pm\sqrt{\lambda^2+\kappa^{-2}}$.
The normalization of $A$ and the bound $\kappa\ge\norm{A^{-1}}$ give
$1/\kappa\le|\lambda|\le1$. Thus every nonzero eigenvalue of $H$ has
absolute value at least $\sqrt2/\kappa$, while every eigenvalue has absolute
value at most $\sqrt{1+\kappa^{-2}}$. Taking the union over an eigenbasis of
$A$ proves the spectral bounds.
\end{proof}

\begin{proposition}[Block-encoding of the auxiliary matrix]\label[proposition]{prop:graph-encoding}
There is a Hermitian unitary $U_H$ on $\rB\rG\rD$, with $a+O(1)$ signal qubits in $\rB$, such that
\begin{equation}\label{eq:graph-encoding}
 \left(\bra0_{\rB}\otimes I_{\rG\rD}\right)U_H
 \left(\ket0_{\rB}\otimes I_{\rG\rD}\right)=\frac{H}{\alpha_H},
 \qquad \alpha_H=1+\kappa^{-1}.
\end{equation}
Each allowed call to $U_H$ uses $O(1)$ matrix queries and $O(1)$ extra gates.
\end{proposition}

\begin{proof}
Hermitianizing $U_A$ as in the proof of
\cref{prop:hermitian-reduction} gives a Hermitian unitary block-encoding of
$(A+A^\dagger)/2=A$ using one additional signal qubit. For
$X_{0j}=(\ket0\bra j+\ket j\bra0)_{\rG}$, a constant-size Hermitian unitary
dilation encodes $X_{0j}$. Hence tensor products give Hermitian unitary
block-encodings of $X_{01}\otimes A$ and
$-X_{02}\otimes I_{\rD}$.

Applying the two-term linear-combination construction
\cite[Lemma 52]{GSLW19} with weights $1$ and $\kappa^{-1}$ gives a Hermitian
unitary $U_H$ whose signal block is
\[
 \frac{X_{01}\otimes A-\kappa^{-1}X_{02}\otimes I_{\rD}}
      {1+\kappa^{-1}}=\frac{H}{\alpha_H}.
\]
The construction adds $O(1)$ signal qubits and uses $O(1)$ matrix queries and
$O(1)$ extra gates.
\end{proof}

\subsection{Ideal overlap}

We first realize the reflection about $\ker H$, then combine it with the
reflection about $e$ and a fractional-linear transducer on two copies of the
public space.

\subsubsection{A canonical transducer for reflection about the kernel}\label{sec:local-transducer}

Let $\mathcal H_0=\mathcal H_{\rG\rD}$ and
$\mathcal K=\mathcal H_{\rB\rG\rD}$. The transducer acts on the direct sum
$\mathcal H_0\oplus\mathcal K_1\oplus\mathcal K_2$, where
$\mathcal K_1,\mathcal K_2$ are two copies of $\mathcal K$. A separate label
register distinguishes these three spaces, which use one shared
$\rB\rG\rD$ workspace as follows:
\[
\begin{array}{c|c|c}
 \text{space} & \text{role} & \text{state in the shared workspace}\\ \hline
 \mathcal H_0 & \text{public}
   & \ket0_{\rB}\otimes\xi,\quad \xi\in\mathcal H_{\rG\rD}\\
 \mathcal K_1 & \text{first private query space}
   & \phi_1\in\mathcal H_{\rB\rG\rD}\\
 \mathcal K_2 & \text{second private query space}
   & \phi_2\in\mathcal H_{\rB\rG\rD}
\end{array}
\]
Thus $\rB$ is the signal register of $U_H$, not the label distinguishing the
spaces. It is fixed to $\ket0_{\rB}$ in the public space and is unrestricted
in either private space. Write
\[
 Q=\proj0_{\rB}\otimes I_{\rG\rD}
\]
for the projector onto the zero-signal subspace of either private space.

\begin{proposition}[Kernel-reflection transducer]\label[proposition]{prop:kernel-transducer}
Let $R_P=2P-I$. For every known $\tau>0$, there is a canonical transducer
\[
 S_P(\tau):=S_P^\circ(\tau)\left(I\oplus U_H\oplus U_H\right)
\]
on $\mathcal H_0\oplus\mathcal K_1\oplus\mathcal K_2$. Its transduction
action on the public register $\rG\rD$ is
$S_P(\tau)\DownTransduce_{\mathcal H_0}=R_P$.
For each $\xi\in\mathcal H_0$, one can choose catalyst components
$\omega_{i,\xi}\in\mathcal K_i$ with total $U_H$-query complexity
\begin{equation}\label{eq:local-weight}
 L_H\left(\xi\right):=\norm{\omega_{1,\xi}}^2+\norm{\omega_{2,\xi}}^2
 =2\alpha_H\left(\tau\norm{P\xi}^2+
                                \frac{\norm{H^+\xi}^2}{\tau}\right).
\end{equation}
A controlled implementation of $S_P^\circ(\tau)$ uses $O(a+1)$ extra gates.
\end{proposition}

\begin{proof}
Put $\mu=\alpha_H\tau$. With respect to the decomposition into the public
space $\mathcal H_0$ and the two private spaces $\mathcal K_1,\mathcal K_2$,
define $S_P^\circ(\tau)$ by
\begin{equation}\label{eq:local-work}
 S_P^\circ(\tau)=
 \begin{pmatrix}
  \dfrac{1-\mu}{1+\mu}I_{\mathcal H_0} & 0 &
      \dfrac{2\sqrt\mu}{1+\mu}
       \left(\bra0_{\rB}\otimes I_{\rG\rD}\right)\\[2mm]
  \dfrac{2\sqrt\mu}{1+\mu}
       \left(\ket0_{\rB}\otimes I_{\rG\rD}\right) & 0 &
      \dfrac{2\mu}{1+\mu}Q-I_{\mathcal K}\\[2mm]
  0 & I_{\mathcal K}-2Q & 0
 \end{pmatrix}.
\end{equation}
For $i=1,2$, the projector $Q$ gives the orthogonal decomposition
\[
 \mathcal K_i
 =Q\mathcal K_i\oplus(I_{\mathcal K}-Q)\mathcal K_i,
 \qquad
 Q\mathcal K_i
 =\left\{\ket0_{\rB}\otimes\xi:\xi\in\mathcal H_0\right\}.
\]
Identifying $\xi\in\mathcal H_0$ with
$\ket0_{\rB}\otimes\xi\in Q\mathcal K_i$, the restriction from
$\mathcal H_0\oplus Q\mathcal K_2$ to
$\mathcal H_0\oplus Q\mathcal K_1$ is
\[
 M_\mu:=\frac1{1+\mu}
 \begin{pmatrix}1-\mu&2\sqrt\mu\\
                 2\sqrt\mu&\mu-1\end{pmatrix}.
\]
The remaining parts act by $I_{\mathcal K}-2Q$ from $\mathcal K_1$ to
$\mathcal K_2$ and by $-I$ from
$(I_{\mathcal K}-Q)\mathcal K_2$ to
$(I_{\mathcal K}-Q)\mathcal K_1$. These input and output subspaces are
mutually orthogonal. Since $M_\mu^2=I$, all three actions are unitary, and so
is $S_P^\circ(\tau)$.

Choose the catalyst components
\begin{equation}\label{eq:local-catalyst}
 \begin{split}
 \omega_{1,\xi}
   &=\sqrt{\alpha_H\tau}\left(\ket0_{\rB}\otimes P\xi\right)
     +\sqrt{\alpha_H/\tau}\,
       U_H\left(\ket0_{\rB}\otimes H^+\xi\right),\\
 \omega_{2,\xi}
   &=\sqrt{\alpha_H\tau}\,
       U_H\left(\ket0_{\rB}\otimes P\xi\right)
     -\sqrt{\alpha_H/\tau}\left(\ket0_{\rB}\otimes H^+\xi\right).
 \end{split}
\end{equation}
Using $U_H^2=I$, $HH^+=I-P$, $HP=0$, and
\eqref{eq:graph-encoding}, we obtain
\begin{align*}
 \left(\bra0_{\rB}\otimes I_{\rG\rD}\right)
 U_H\left(\ket0_{\rB}\otimes P\right)
 &=\alpha_H^{-1}HP=0,\\
 \left(\bra0_{\rB}\otimes I_{\rG\rD}\right)
 \left(\ket0_{\rB}\otimes H^+\right)
 &=H^+,\\
 \left(\bra0_{\rB}\otimes I_{\rG\rD}\right)U_H\omega_{2,\xi}
 &=\sqrt\mu\,P\xi-\frac1{\sqrt\mu}(I-P)\xi,\\
 \frac{1-\mu}{1+\mu}\xi
 +\frac{2\sqrt\mu}{1+\mu}
   \left(\bra0_{\rB}\otimes I_{\rG\rD}\right)U_H\omega_{2,\xi}
 &=P\xi-(I-P)\xi=R_P\xi,\\
 \frac{2\sqrt\mu}{1+\mu}\left(\ket0_{\rB}\otimes\xi\right)
 +\left(\frac{2\mu}{1+\mu}Q-I_{\mathcal K}\right)U_H\omega_{2,\xi}
 &=2\sqrt\mu\left(\ket0_{\rB}\otimes P\xi\right)
   -U_H\omega_{2,\xi}=\omega_{1,\xi},\\
 \left(I_{\mathcal K}-2Q\right)U_H\omega_{1,\xi}
 &=\sqrt\mu\,U_H\left(\ket0_{\rB}\otimes P\xi\right)\\
 &\quad-\sqrt{\alpha_H/\tau}
   \left(\ket0_{\rB}\otimes H^+\xi\right)
   =\omega_{2,\xi}.
\end{align*}
Therefore
\begin{equation}\label{eq:local-identity}
 S_P^\circ\left(\tau\right)\left(I\oplus U_H\oplus U_H\right)
 \left(\xi\oplus\omega_{1,\xi}\oplus\omega_{2,\xi}\right)
 =R_P\xi\oplus\omega_{1,\xi}\oplus\omega_{2,\xi}.
\end{equation}
Thus $\omega_{1,\xi},\omega_{2,\xi}$ are catalyst witnesses for the transduction action
$R_P$. Moreover,
\[
 \ip{\ket0_{\rB}\otimes P\xi}
     {U_H\left(\ket0_{\rB}\otimes H^+\xi\right)}
 =\alpha_H^{-1}\ip{P\xi}{HH^+\xi}
 =\alpha_H^{-1}\ip{P\xi}{\left(I-P\right)\xi}=0.
\]
Since $U_H$ is Hermitian, the other cross term vanishes as well. Therefore
\[
 \norm{\omega_{1,\xi}}^2=\norm{\omega_{2,\xi}}^2
 =\alpha_H\tau\norm{P\xi}^2
   +\frac{\alpha_H}{\tau}\norm{H^+\xi}^2,
\]
which proves \eqref{eq:local-weight}. Finally, in the shared-workspace
encoding described above, the factors $\ket0_{\rB}$, $\bra0_{\rB}$, and $Q$
in \eqref{eq:local-work} only indicate transitions between the labeled spaces and
a test of $\rB=0$. Thus the work unitary uses a constant number of label
rotations and signs, plus controls on the $a+O(1)$ signal qubits in $\rB$.
Its controlled implementation therefore uses $O(a+1)$ extra gates.
\end{proof}

\begin{proposition}[Input-reflection transducer]\label[proposition]{prop:input-transducer}
Let $R_e=2\proj e-I$. There is a canonical transducer
\[
 S_e=S_e^\circ\left(I\oplus R_e\right)
\]
with transduction action $S_e\DownTransduce_{\rG\rD}=R_e$. For each public
vector $\xi$, it has a catalyst with $R_e$-query complexity
$L_e(\xi)=\norm\xi^2$. A controlled call to $R_e$ uses two vector
queries and $O(n+1)$ extra gates, while a controlled $S_e^\circ$ uses $O(1)$
extra gates.
\end{proposition}

\begin{proof}
Since $\ket e=\ket1_{\rG}\ket b_{\rD}$,
\[
 R_e=\left(I_{\rG}\otimes U_b\right)
       \left(2\proj1_{\rG}\otimes\proj0_{\rD}-I_{\rG\rD}\right)
       \left(I_{\rG}\otimes U_b^\dagger\right),
\]
which gives the stated implementation cost. Let $S_e^\circ$ exchange the
public and private labels. Then
\[
 S_e\left(\xi\oplus\xi\right)
 =S_e^\circ\left(\xi\oplus R_e\xi\right)
 =R_e\xi\oplus\xi.
\]
Thus $\xi$ is a catalyst and $L_e(\xi)=\norm\xi^2$.
\end{proof}

\subsubsection{Increasing the kernel overlap}

We first obtain constant overlap by a fractional-linear function of the two-reflection unitary. We then compose the local transducers to bound both query costs with one catalyst. The finite-circuit invariant and the preparation guarantee are proved in \cref{sec:finite}.

\begin{proposition}[Fractional-linear transducer]\label[proposition]{prop:transducer}
Let $U$ be unitary on a finite-dimensional space $\cH$ and let $r\in(-1,1)$. On $\cH\oplus\cH$, the first copy public and the second private, define
\begin{equation}\label{eq:transducer-def}
 B_r=\begin{pmatrix}-rI&\sqrt{1-r^2}I\\
                  \sqrt{1-r^2}I&rI\end{pmatrix}.
\end{equation}
The canonical transducer $B_r(I\oplus U)$ has unitary transduction action $(U-rI)(I-rU)^{-1}$. For every $\xi\in\cH$, the private vector $q=\sqrt{1-r^2}(I-rU)^{-1}\xi$ satisfies
\begin{equation}\label{eq:transducer-identity}
 B_r\left(I\oplus U\right)\left(\xi\oplus q\right)
   =\left(U-rI\right)\left(I-rU\right)^{-1}\xi\oplus q.
\end{equation}
\end{proposition}

\begin{proof}
Since $\norm{rU}=|r|<1$, the Neumann series $\sum_{j=0}^{\infty}r^jU^j$ converges in operator norm and is the inverse of $I-rU$. Thus $q$ is well defined. Multiplying its definition by $I-rU$ gives
\[
 (I-rU)q=\sqrt{1-r^2}\,\xi,
 \qquad\text{or equivalently}\qquad
 q=\sqrt{1-r^2}\,\xi+rUq.
\]
After $I\oplus U$ is applied, the input to $B_r$ is $\xi\oplus Uq$.
By \eqref{eq:transducer-def}, the second component of the output is therefore
$\sqrt{1-r^2}\,\xi+rUq=q$. The first component equals
\begin{align*}
 -r\xi+\sqrt{1-r^2}\,Uq
 &=\bigl[-r\left(I-rU\right)+\left(1-r^2\right)U\bigr]\left(I-rU\right)^{-1}\xi\\
 &=\left(U-rI\right)\left(I-rU\right)^{-1}\xi.
\end{align*}
Direct multiplication shows that $B_r^\dagger=B_r$ and $B_r^2=I$. To check the public action, take an eigenvalue $z$ of $U$, with $|z|=1$. Both $|z-r|^2$ and $|1-rz|^2$ equal $1+r^2-2r\operatorname{Re}z$. Their ratio has modulus one, and the spectral theorem shows that $(U-rI)(I-rU)^{-1}$ is unitary. 
\end{proof}

\paragraph{Action on two reflections.}

The next lemma needs only the projection bounds established above. We later apply it to $P=\Pi_{\ker H}$ and $e=(0,b,0)$.

\begin{lemma}[Two-reflection overlap and catalyst norm]\label[lemma]{lem:transducer-geometry}
Let $P$ be an orthogonal projector and $e$ a unit vector. Suppose $\kappa\ge2$, $1\le s\le\kappa$, and $3s/8\le\widehat s\le5s/2$ satisfy
\[
 \frac{s^2}{2\kappa^2}\le\norm{Pe}^2
 \le\min\!\left\{\frac{s^2}{\kappa^2},\frac12\right\}.
\]
Set $U=-(2\proj e-I)(2P-I)$ and
\begin{equation}\label{eq:mixing-parameter}
 r=\frac{\kappa-16\widehat s}{\kappa+16\widehat s}.
\end{equation}
Then $r\in(-1,1)$, and the transduction output of $B_r(I\oplus U)$ on $e$
has real overlap with $Pe/\norm{Pe}$ bounded below by an absolute positive
constant. Its catalyst in \cref{prop:transducer} lies in
$\operatorname{span}_{\R}\{Pe,(I-P)e\}$ and has squared norm
$O(\kappa/\widehat s)$.
\end{lemma}

\begin{proof}
Since $\widehat s>0$, we have $|r|<1$. Put
\[
 \gamma=\norm{Pe},\qquad
 u=\frac{Pe}{\gamma},\qquad
 w=\frac{(I-P)e}{\sqrt{1-\gamma^2}},\qquad
 \sin\theta=\gamma.
\]
Then $0<\theta\le\pi/4$, $e=\sin\theta\,u+\cos\theta\,w$, and on the
real plane $\operatorname{span}_{\R}\{u,w\}$,
\[
 U=\begin{pmatrix}\cos2\theta&\sin2\theta\\
                  -\sin2\theta&\cos2\theta\end{pmatrix}.
\]
Thus this plane is invariant under the real operators $U$ and
$(I-rU)^{-1}$. In particular, the catalyst
$q=\sqrt{1-r^2}(I-rU)^{-1}e$ lies in the claimed real span.

Set
\[
 t=\frac{1+r}{1-r}\tan\theta
   =\frac{\kappa\tan\theta}{16\widehat s},
 \qquad \varphi=2\arctan t.
\]
For an eigenvalue $e^{2i\theta}$ of $U$,
\[
 \frac{e^{2i\theta}-r}{1-re^{2i\theta}}
 =\frac{1+it}{1-it}=e^{i\varphi}.
\]
Hence the transduction output is
$\psi=\sin(\theta+\varphi)u+\cos(\theta+\varphi)w$.
The hypotheses give
\[
 \frac{s}{\sqrt2\kappa}\le\tan\theta\le\frac{\sqrt2s}{\kappa},
 \qquad
 \frac1{40\sqrt2}\le t\le\frac{\sqrt2}{6}<\tan\frac\pi8.
\]
In particular, $t^2\le1/18$, $0<\varphi<\pi/4$, and
$\theta+\varphi<\pi/2$. Therefore
\[
 \ip u\psi
 \ge\sin\varphi
 =\frac{2t}{1+t^2}
 \ge\frac{36}{19}\cdot\frac1{40\sqrt2}
 =\frac9{190\sqrt2}>\frac1{30}.
\]

Finally, $U+U^\dagger=2\cos(2\theta)I$ on this plane, so
\begin{align*}
 \norm q^2
 &=\frac{1-r^2}{1+r^2-2r\cos2\theta}\\
 &=\frac{16\kappa\widehat s}
         {\left(16\widehat s\right)^2\cos^2\theta+\kappa^2\sin^2\theta}
 \le\frac{\kappa}{8\widehat s}\le\frac{\kappa}{3s}.
\end{align*}
Here the equality substitutes $r$, while the inequalities use
$\cos^2\theta\ge1/2$ and $\widehat s\ge3s/8$.
\end{proof}

\subsection{Preparation transducer and separate oracle costs}\label{sec:register-implementation}

The fractional-linear transducer of \cref{prop:transducer} requires the
unitary $U=-R_eR_P$. We realize its two reflections separately: the
kernel-reflection transducer of \cref{prop:kernel-transducer} uses $U_H$,
whereas the input-reflection transducer of \cref{prop:input-transducer} uses
$R_e$.

Set $\tau=\widehat s$, choose $r$ as in \eqref{eq:mixing-parameter}, and write
\[
 U=-R_eR_P,\qquad
 q=\sqrt{1-r^2}(I-rU)^{-1}e,\qquad
 \psi=(U-rI)(I-rU)^{-1}e.
\]
Let $\omega_{1,q},\omega_{2,q}$ be the local catalysts from
\eqref{eq:local-catalyst}, with $\tau=\widehat s$, and set
\begin{equation}\label{eq:full-catalyst}
 \omega=q\oplus\omega_{1,q}\oplus\omega_{2,q}\oplus R_Pq.
\end{equation}

A three-qubit register $\rS$ labels the public space and the four private
spaces. For the augmented vector $e\oplus\omega$, their roles are
\begin{equation}\label{eq:space-labels}
\begin{array}{c|l|c|c}
 \rS & \text{role} & \text{component} & \text{oracle action}\\ \hline
 0 & \text{public} & e & I\\
 1 & \text{non-query private} & q & I\\
 2 & \text{first }U_H\text{-query part} & \omega_{1,q} & U_H\\
 3 & \text{second }U_H\text{-query part} & \omega_{2,q} & U_H\\
 4 & R_e\text{-query part} & R_Pq & R_e
\end{array}
\end{equation}
The labels $5,6,7$ are unused.  All circuit operators below act as the
identity on their span; the direct-sum notation records only the five active
spaces.  All active spaces share $\rG\rD$. The signal register $\rB$ is
unrestricted for $\rS=2,3$ and is fixed to $\ket0_{\rB}$ otherwise.

In the order of \eqref{eq:space-labels}, the oracle layer is
\begin{equation}\label{eq:oracle-layer}
\begin{aligned}
 \mathcal O
 &=I\oplus I\oplus U_H\oplus U_H\oplus R_e\\
 &=\left(I_{\rS}-\sum_{j=2}^4\proj j_{\rS}\right)
      \otimes I_{\rB\rG\rD}
   +\left(\proj2_{\rS}+\proj3_{\rS}\right)\otimes U_H\\
 &\quad+\proj4_{\rS}\otimes\left(I_{\rB}\otimes R_e\right).
\end{aligned}
\end{equation}

Define the label unitaries
\[
\begin{aligned}
 \widetilde B_{r,\rS}
 &=-r\proj0_{\rS}
   +\sqrt{1-r^2}\left(\ket0\!\bra1+\ket1\!\bra0\right)_{\rS}
   +r\proj1_{\rS}+\sum_{j=2}^7\proj j_{\rS},\\
 X_{\rS}
 &=\left(\ket1\!\bra4+\ket4\!\bra1\right)_{\rS}
   +\sum_{j\in\{0,2,3,5,6,7\}}\proj j_{\rS},
 \qquad Z_{\rS}=I-2\proj1_{\rS}.
\end{aligned}
\]
In the order of \eqref{eq:space-labels}, also set
$\widetilde S_P^\circ(\widehat s)=I\oplus S_P^\circ(\widehat s)\oplus I$.
The work unitary and the resulting canonical transducer are
\begin{equation}\label{eq:transducer-work}
\begin{aligned}
 S^\circ_{\rS\rB}
 &=\left(\widetilde B_{r,\rS}\otimes I_{\rB}\right)
   \left(Z_{\rS}\otimes I_{\rB}\right)
   \left(X_{\rS}\otimes I_{\rB}\right)
   \widetilde S_P^\circ(\widehat s),\\
 S^\circ
 &=S^\circ_{\rS\rB}\otimes I_{\rG\rD},\\
 S
 &=S^\circ\mathcal O
  =S^\circ\left(I\oplus I\oplus U_H\oplus U_H\oplus R_e\right).
\end{aligned}
\end{equation}
The factors act from right to left. The first factor applied uses the
signal-controlled extension from \cref{prop:kernel-transducer}.

\begin{proposition}[Preparation transducer]\label[proposition]{prop:preparation-transducer}
The transducer $S$ in \eqref{eq:transducer-work} satisfies
\begin{equation}\label{eq:preparation-transduction}
 S(e\oplus\omega)=\psi\oplus\omega,
 \qquad \norm\psi=1,
 \qquad \ip u\psi>\frac1{30}.
\end{equation}
Its catalyst costs obey
\begin{equation}\label{eq:transducer-cost}
 L_H:=\norm{\omega_{1,q}}^2+\norm{\omega_{2,q}}^2<8\kappa,
 \qquad
 L_e:=\norm{R_Pq}^2\le\frac{\kappa}{8\widehat s}
        \le\frac{\kappa}{3s},
 \qquad
 W:=\norm\omega^2<9\kappa.
\end{equation}
A controlled implementation of $S^\circ$ uses $O(a+1)$ extra gates.
\end{proposition}

\begin{proof}
Write $z=R_Pq$. Starting from
$e\oplus\omega$, the oracle and work operations act as follows:
\begin{align*}
 e\oplus q\oplus\omega_{1,q}\oplus\omega_{2,q}\oplus z
 &\xrightarrow{\mathcal O}
 e\oplus q\oplus U_H\omega_{1,q}\oplus U_H\omega_{2,q}\oplus R_ez\\
 &\xrightarrow{S_P^\circ(\widehat s)}
 e\oplus z\oplus\omega_{1,q}\oplus\omega_{2,q}\oplus R_ez\\
 &\xrightarrow{1\leftrightarrow4}
 e\oplus R_ez\oplus\omega_{1,q}\oplus\omega_{2,q}\oplus z\\
 &\xrightarrow{\text{sign on }1}
 e\oplus Uq\oplus\omega_{1,q}\oplus\omega_{2,q}\oplus z\\
 &\xrightarrow{B_r}
 \psi\oplus q\oplus\omega_{1,q}\oplus\omega_{2,q}\oplus z.
\end{align*}
The second line is \eqref{eq:local-identity}; the fourth uses
$-R_ez=-R_eR_Pq=Uq$; and the last is
\eqref{eq:transducer-identity}. Thus every private component is restored,
which proves the transduction identity. The overlap bound follows from
\cref{lem:transducer-geometry}.

It remains to bound the catalyst. Put
$w=(I-P)e/\norm{(I-P)e}$. The projection bounds in
\cref{lem:geometry} give $\norm{(I-P)e}^2\ge1/2$, and hence
\[
 \norm{H^+w}^2
 =\frac{\norm{H^+e}^2}{1-\norm{Pe}^2}\le2s^2.
\]
By \cref{lem:transducer-geometry}, $q$ lies in the real span of the
orthonormal pair $u,w$. Therefore
$Pq=u\langle u,q\rangle$ and
$H^+q=H^+w\langle w,q\rangle$. Substituting these identities and
$\tau=\widehat s$ into \eqref{eq:local-weight} gives
\begin{align*}
 L_H
 &=2\alpha_H\left(\widehat s\,|\langle u,q\rangle|^2
      +\frac{\norm{H^+w}^2}{\widehat s}|\langle w,q\rangle|^2\right)\\
 &\le4\left(\frac{5s}{2}|\langle u,q\rangle|^2
      +\frac{16s}{3}|\langle w,q\rangle|^2\right)
 \le\frac{64s}{3}\norm q^2
 \le\frac{64}{9}\kappa<8\kappa.
\end{align*}
Here we used $\alpha_H\le2$, the estimate bounds, and
$\norm q^2\le\kappa/(8\widehat s)\le\kappa/(3s)$ from
\cref{lem:transducer-geometry}. Since $R_P$ is unitary,
\[
 L_e=\norm{R_Pq}^2=\norm q^2\le\frac{\kappa}{8\widehat s},
 \qquad
 W=L_H+2\norm q^2<9\kappa.
\]

Finally, the signal-controlled $S_P^\circ(\widehat s)$ costs $O(a+1)$
gates. The swap, sign, and $B_r$ act on the constant-size label register, so
an external control preserves the $O(a+1)$ bound.
\end{proof}

\paragraph{Finite implementation.}
Let the resulting circuit from \cref{prop:finite-preparation} be
\begin{equation}\label{eq:preparation-circuit}
 U_{\rm fin}=\Finite(S;\kappa,\widehat s).
\end{equation}
It implements the public action of $S$ with error less than $10^{-3}$ and
uses $O(\kappa)$ controlled calls
to $U_H$ and $O(\kappa/s)$ controlled calls to $R_e$. By
\cref{prop:graph-encoding,prop:input-transducer}, these become
$O(\kappa)$ matrix queries, $O(\kappa/s)$ vector queries, and
\[
 O\!\left(\kappa(a+1)+\frac\kappa s(n+1)\right)
\]
extra one- and two-qubit gates.

\subsection{Exact direction under finite implementation}\label{sec:finite}

\begin{lemma}[Exact kernel alignment]\label[lemma]{lem:alignment}
Let
\[
 P=\Pi_{\ker H},\qquad
 \ket e_{\rG\rD}=\ket1_{\rG}\ket b_{\rD},\qquad
 \ket u_{\rG\rD}=\frac{P\ket e}{\norm{P\ket e}}.
\]
Let the auxiliary register and its all-zero state be
\[
 \rA=\rS\rB\rT\rE,
 \qquad
 \ket0_{\rA}=\ket0_{\rS}\ket0_{\rB}\ket0_{\rT}\ket0_{\rE},
\]
where $\rS,\rB,\rT,\rE$ are respectively the label, signal, clock, and
ancilla registers.
Let $S$ be the preparation transducer in \eqref{eq:transducer-work}, and let
$U_{\rm fin}=\Finite(S;\kappa,\widehat s)$ be its real finite implementation.
Define
\[
 \ket\Psi_{\rA\rG\rD}
 =U_{\rm fin}\left(\ket0_{\rA}\ket e_{\rG\rD}\right).
\]
Then there is a real number $\beta$ such that
\begin{equation}\label{eq:exact-alignment}
 \left(\bra0_{\rA}\otimes P\right)\ket\Psi
 =\beta\ket u.
\end{equation}
\end{lemma}
\begin{proof}
Let $y=(\bra0_{\rA}\otimes I_{\rG\rD})\ket\Psi$, and set
\begin{equation}\label{eq:krylov}
 \cK=\spanR\{H^je:j\ge0\},\qquad
 \cM=\left(\ket0_{\rB}\otimes\cK\right)+U_H\left(\ket0_{\rB}\otimes\cK\right).
\end{equation}
It is enough to prove $y\in\cK$. Indeed, writing
$y=\sum_jc_jH^je$ with $c_j\in\R$ would give
\[
 Py=c_0Pe=c_0\norm{Pe}\,u
\]
because $PH=0$, proving the claim with $\beta=c_0\norm{Pe}$.

The proof is based on the following three invariance relations:
\begin{equation}\label{eq:space-identities}
 U_H\cM=\cM,\qquad
 \proj0_{\rB}\cM\subseteq\ket0_{\rB}\otimes\cK,\qquad
 R_e\cK=\cK.
\end{equation}
The first follows from $U_H^2=I$. For the second, $\cK$ is $H$-invariant
and the block encoding gives
\[
 \proj0_{\rB}U_H\left(\ket0_{\rB}\otimes\cK\right)
 =\ket0_{\rB}\otimes\left(H/\alpha_H\right)\cK
 \subseteq\ket0_{\rB}\otimes\cK.
\]
Finally, if $z=\sum_jc_jH^je\in\cK$, then
$\langle e,z\rangle\in\R$ and
$R_ez=2e\langle e,z\rangle-z\in\cK$. Since $R_e^2=I$, the last inclusion is
an equality.

Let $\mathcal L$ consist of all real vectors of the form
\begin{equation}\label{eq:assigned-space}
\begin{aligned}
 \mathcal L=\Bigl\{&
   \sum_{j\in\{0,1,4\}}\ket j_{\rS}\ket0_{\rB}\otimes k_j
   +\ket2_{\rS}\otimes m_2+\ket3_{\rS}\otimes m_3:\\
 & k_j\in\cK,\quad m_2,m_3\in\cM\Bigr\}.
\end{aligned}
\end{equation}
Using \eqref{eq:space-identities}, the definition of $\cM$, and the block
formula \eqref{eq:local-work}, the invariance checks are
\[
\begin{aligned}
 U_H\cM=\cM,\quad R_e\cK=\cK
 &\quad\Longrightarrow\quad \mathcal O\mathcal L=\mathcal L,\\
 (\bra0_{\rB}\otimes I_{\rG\rD})\cM\subseteq\cK,\quad
 (\ket0_{\rB}\otimes I_{\rG\rD})\cK\subseteq\cM
 &\quad\Longrightarrow\quad
 \left(\widetilde S_P^\circ(\widehat s)\otimes I_{\rG\rD}\right)
 \mathcal L=\mathcal L,\\
 \left(X_{\rS}\otimes I_{\rB\rG\rD}\right)\mathcal L
 &=\left(Z_{\rS}\otimes I_{\rB\rG\rD}\right)\mathcal L\\
 &=\left(\widetilde B_{r,\rS}\otimes I_{\rB\rG\rD}\right)\mathcal L
 =\mathcal L.
\end{aligned}
\]
Together with \eqref{eq:transducer-work}, these relations give
$S^\circ\mathcal L=\mathcal L$.

Tensor $\mathcal L$ with the real span of the computational basis of
$\rT\rE$. The controlled calls in $\Finite(S;\kappa,\widehat s)$ preserve
this enlarged space, as do the real clock and control operations of
\cref{lem:real-compiler}. Since the input is $\ket0_{\rA}\ket e_{\rG\rD}$,
projecting the output onto $\ket0_{\rA}$ therefore gives
\[
 y=\left(\bra0_{\rA}\otimes I_{\rG\rD}\right)\ket\Psi\in\cK.
\]
This completes the proof. 
\end{proof}


We now prove the main result in this section.

\begin{proof}[Proof of \Cref{prop:coarse}]
The circuit in \cref{alg:coarse} prepares $\ket1_{\rG}\ket b_{\rD}$ on
$\rG\rD$, with all auxiliary registers initialized to zero, and applies
$U_{\rm fin}$ from \eqref{eq:preparation-circuit}. Write
$\psi$ for its ideal public output. By
\cref{lem:transducer-geometry}, $\langle u,\psi\rangle>1/30$. The real
implementation in \cref{prop:finite-preparation} has error less than
$10^{-3}$ and returns all registers coherently.

The output is a unit vector, and contraction with $\bra0_{\rA}$ gives
\[
 \left\|\left(\bra0_{\rA}\otimes I_{\rG\rD}\right)\Psi\right\|\le1,
 \qquad
 \left\|\left(\bra0_{\rA}\otimes I_{\rG\rD}\right)\Psi-\psi\right\|
 <10^{-3}.
\]
By \cref{lem:alignment},
$\left(\bra0_{\rA}\otimes P\right)\ket\Psi=\beta\ket u$ for real
$\beta$. Consequently
\[
 \beta
 =\left\langle u,\left(\bra0_{\rA}\otimes I_{\rG\rD}\right)\Psi\right\rangle
 \ge\langle u,\psi\rangle
 -\left\|\left(\bra0_{\rA}\otimes I_{\rG\rD}\right)\Psi-\psi\right\|
 >\frac1{30}-10^{-3}>\frac1{32}.
\]
The upper bound follows from
$|\beta|\le\left\|\left(\bra0_{\rA}\otimes I_{\rG\rD}\right)\Psi\right\|
\le1$.
Define
\[
 \ket h_{\rG\rD}
 =\left(\bra0_{\rA}\otimes\left(I_{\rG\rD}-P\right)\right)
   \ket\Psi_{\rA\rG\rD}.
\]
Since $P^2=P$, we have $P\ket h=0$, and
\[
 \left(\bra0_{\rA}\otimes I_{\rG\rD}\right)\ket\Psi
 =\left(\bra0_{\rA}\otimes P\right)\ket\Psi+\ket h
 =\beta\ket u+\ket h.
\]
This proves \eqref{eq:coarse-guarantee}.

By \cref{thm:bjy,prop:preparation-transducer}, the controlled $S^\circ$
circuits and compiler overhead contribute $O(\kappa(a+1))$ extra gates.  The
$O(\kappa)$ controlled
queries to $U_H$ use $O(\kappa)$ matrix queries and extra gates by
\cref{prop:graph-encoding}.  The
$O(1+\kappa/\widehat s)=O(\kappa/s)$ controlled queries to $R_e$ use
$O(\kappa/s)$ vector queries and $O((\kappa/s)(n+1))$ extra gates by
\cref{prop:input-transducer}. The initial preparation uses one vector query,
absorbed since $\kappa/s\ge1$. Finally, $\rS\rB$ uses $O(a+1)$ qubits, $\rT$
uses $O(\log\kappa)$, and $\rE$ uses $O(a+\log\kappa)$ and returns to zero.
These give all stated costs.
\end{proof}

\section{Matrix-only refinement}\label{sec:refinement}


\subsection{Polynomial eigenstate filtering}

\begin{lemma}[Kernel-filter polynomial, {\cite[Lemma 2]{LinTong20}}]\label[lemma]{lem:filter-polynomial}
For $0<\delta\le1/\sqrt{12}$ and $0<\eta<1/2$, set
\[
 \ell=\left\lceil\frac{\log(2/\eta)}{\sqrt2\delta}\right\rceil
\]
and define
\begin{equation}\label{eq:chebyshev-filter}
 R_{\delta,\eta}(x)=
 \frac{T_\ell\!\left(-1+2\dfrac{x^2-\delta^2}{1-\delta^2}\right)}
 {T_\ell\!\left(-\dfrac{1+\delta^2}{1-\delta^2}\right)},
\end{equation}
where $T_\ell(\cos\theta)=\cos(\ell\theta)$ is the Chebyshev polynomial of
the first kind. Then $R_{\delta,\eta}$ is a real even polynomial of degree
$2\ell=O(\delta^{-1}\log(1/\eta))$. It satisfies
$R_{\delta,\eta}(0)=1$, $|R_{\delta,\eta}(x)|\le1$ for $|x|\le1$, and
$|R_{\delta,\eta}(x)|\le\eta$ for $\delta\le|x|\le1$.
\end{lemma}

\begin{lemma}[QSVT kernel filter]\label[lemma]{lem:filter}
Set $\alpha_H=1+\kappa^{-1}$ and
$\delta=\min\{1/(\kappa\alpha_H),1/\sqrt{12}\}$. For every
$0<\eta<1/2$, let $R_{\delta,\eta}$ be the polynomial from
\cref{lem:filter-polynomial}. There is a unitary
$\textup{\FilterKernel}(U_A;\kappa,\eta)_{\rF\rG\rD}$ such that
\begin{equation}\label{eq:filter-register-block}
 \left(\bra0_{\rF}\otimes I_{\rG\rD}\right)
 \textup{\FilterKernel}\left(U_A;\kappa,\eta\right)_{\rF\rG\rD}
 \left(\ket0_{\rF}\otimes I_{\rG\rD}\right)
 =R_{\delta,\eta}\left(H/\alpha_H\right)
\end{equation}
and
\begin{equation}\label{eq:filter-operator-bound}
 \left\|R_{\delta,\eta}\left(H/\alpha_H\right)-P\right\|\le\eta.
\end{equation}
The circuit uses $O(\kappa\log(1/\eta))$ matrix queries,
$O(\kappa(a+1)\log(1/\eta))$ extra gates, and $a+O(1)$ signal qubits in
$\rF$. It uses no vector queries.
\end{lemma}

\begin{proof}
By \cref{lem:geometry}, the spectrum of $H/\alpha_H$ is contained in
\[
 \{0\}\cup[-1,-\delta]\cup[\delta,1].
\]
The properties in \cref{lem:filter-polynomial} and the spectral theorem
therefore give \eqref{eq:filter-operator-bound}. Since $R_{\delta,\eta}$ is
real, even, and bounded by one on $[-1,1]$, \cref{lem:qsvt} applied to the
block-encoding from \cref{prop:graph-encoding} gives
\eqref{eq:filter-register-block}. Moreover,
$1/\delta=\max\{\kappa\alpha_H,\sqrt{12}\}=O(\kappa)$, so
$\deg R_{\delta,\eta}=O(\kappa\log(1/\eta))$. The stated costs follow from
the QSVT degree and the implementation costs of $U_H$. The signal qubits used
for these calls are assigned to $\rF$, separately from the preparation
register $\rB\subset\rA$.
\end{proof}

\subsection{Polynomial approximation of the correction operator}

\begin{proposition}[Correction identity]\label[proposition]{prop:correction-identity}
Let $\kappa\ge1$ and let $A$ be Hermitian with
$\operatorname{spec}(A)\subseteq[-1,-1/\kappa]\cup[1/\kappa,1]$. Then
\begin{equation}\label{eq:correction-identity}
 I\preceq I+\kappa^{-2}A^{-2}\preceq2I,
 \qquad
 \left(I+\kappa^{-2}A^{-2}\right)A\left(A^2+\kappa^{-2}I\right)^{-1}=A^{-1}.
\end{equation}
\end{proposition}

\begin{proof}
The spectral assumption gives $0\prec\kappa^{-2}A^{-2}\preceq I$, while
\[
 \left(I+\kappa^{-2}A^{-2}\right)A
 =A+\kappa^{-2}A^{-1}
 =A^{-1}\left(A^2+\kappa^{-2}I\right)
\]
proves the identity.
\end{proof}


\begin{lemma}[Correction polynomial]\label[lemma]{lem:correction-polynomial}
For $\kappa\ge2$ and $0<\eta<1/2$, let $p_{1/\kappa,\eta}$ be the polynomial
from \cref{lem:inverse-square}, and define
\begin{equation}\label{eq:correction-poly}
 c_{\kappa,\eta}(x)=\frac{1+2p_{1/\kappa,\eta}(x)}4.
\end{equation}
Then $c_{\kappa,\eta}$ is a real even polynomial of degree
$O(\kappa\log(1/\eta))$. It satisfies $|c_{\kappa,\eta}(x)|\le3/4$ for
$|x|\le1$, and
$|c_{\kappa,\eta}(x)-(1+\kappa^{-2}x^{-2})/4|\le\eta/2$ for
$1/\kappa\le|x|\le1$.
\end{lemma}

\begin{lemma}[QSVT correction]\label[lemma]{lem:correction}
For $\kappa\ge2$ and $0<\eta<1/2$, let $c_{\kappa,\eta}$ be the polynomial
from \cref{lem:correction-polynomial}. There is a unitary
$\textup{\CorrectSolution}(U_A;\kappa,\eta)_{\rC\rD}$ such that
\begin{equation}\label{eq:correction-register-block}
 \left(\bra0_{\rC}\otimes I_{\rD}\right)\,
 \textup{\CorrectSolution}\left(U_A;\kappa,\eta\right)_{\rC\rD}\,
 \left(\ket0_{\rC}\otimes I_{\rD}\right)
 =c_{\kappa,\eta}(A).
\end{equation}
Moreover,
\begin{equation}\label{eq:correction-bound}
 \left\|c_{\kappa,\eta}(A)
 -\frac{I+\kappa^{-2}A^{-2}}4\right\|\le\frac\eta2.
\end{equation}
The circuit uses $O(\kappa\log(1/\eta))$ matrix queries,
$O(\kappa(a+1)\log(1/\eta))$ extra gates, and $a+O(1)$ signal qubits in
$\rC$. It uses no vector queries.
\end{lemma}

\begin{proof}
The input conditions imply
$\operatorname{spec}(A)\subseteq[-1,-1/\kappa]\cup[1/\kappa,1]$. Hence the
scalar approximation in \cref{lem:correction-polynomial} and the spectral
theorem give \eqref{eq:correction-bound}. Since $c_{\kappa,\eta}$ is real,
even, and bounded by one on $[-1,1]$, \cref{lem:qsvt} gives
\eqref{eq:correction-register-block} with the stated costs.
\end{proof}

The preparation register $\rA$, filter register $\rF$, and correction
register $\rC$ remain disjoint until the final test.

\subsection{Coherent composition and output error}

\begin{proposition}[Matrix-only refinement]\label[proposition]{prop:refinement}
Suppose a unit vector $\ket\Psi_{\rA\rG\rD}$ satisfies
\eqref{eq:coarse-guarantee}. For $0<\varepsilon<1/2$, apply
$\textup{\FilterKernel}$ and $\textup{\CorrectSolution}$ with
$\eta=\varepsilon/1024$, followed by the measurement
$\{\Pi_{\rm succ},I-\Pi_{\rm succ}\}$. The acceptance probability exceeds
$1/65536$, and the accepted state in $\rD$ is within distance
$\varepsilon/2$ of $\ket x_{\rD}$. The procedure uses
$O(\kappa\log(1/\varepsilon))$ matrix queries,
$O(\kappa(a+1)\log(1/\varepsilon))$ extra gates, and no vector queries.
\end{proposition}

\begin{proof}
Choose the polynomial accuracy $\eta=\varepsilon/1024$. Use
$\textup{\FilterKernel}$ and $\textup{\CorrectSolution}$ from
\cref{lem:filter,lem:correction}.

\emph{Accepted vector.}
For this proof let $F_{\rF\rG\rD}$ and $C_{\rC\rD}$ denote the filter and correction circuits in \eqref{eq:filter-register-block}--\eqref{eq:correction-register-block}. Tensor factors are placed according to their register subscripts. The unnormalized data vector selected by \eqref{eq:success-projector} is
\begin{align*}
 \ket z_{\rD}
 &=\left(\bra0_{\rA}\bra0_{\rF}\bra0_{\rC}\bra2_{\rG}\otimes I_{\rD}\right)
     C_{\rC\rD}F_{\rF\rG\rD}
     \ket\Psi_{\rA\rG\rD}\ket0_{\rF}\ket0_{\rC}\\
 &=c_{\kappa,\eta}(A)\left(\bra2_{\rG}\otimes I_{\rD}\right)
     R_{\delta,\eta}\left(H/\alpha_H\right)
     \left(\bra0_{\rA}\otimes I_{\rG\rD}\right)
     \ket\Psi_{\rA\rG\rD}.
\end{align*}
Here the contraction on $\rA$ commutes past both circuits because they do not
act on $\rA$; selecting the $\rC=0$ and $\rF=0$ blocks gives the correction
and filter polynomials, respectively. Thus
\begin{equation}\label{eq:refined-vector}
 \Pi_{\rm succ}C_{\rC\rD}F_{\rF\rG\rD}
    \ket\Psi_{\rA\rG\rD}\ket0_{\rF}\ket0_{\rC}
 =\ket0_{\rA}\ket0_{\rF}\ket0_{\rC}\ket2_{\rG}\ket z_{\rD}.
\end{equation}
The success probability is $\norm z^2$, and the conditional data state is $\ket z_{\rD}/\norm z$.

\emph{Error and success probability.}
Put $D=A^2+\kappa^{-2}I$ and $\gamma=\norm{Pe}$ within this proof. Since
$\ket\Psi$ is a unit vector,
\[
 \left\|\left(\bra0_{\rA}\otimes I_{\rG\rD}\right)\Psi\right\|\le1.
\]
Compare
\[
\begin{aligned}
 \ket y_{\rD}
 &=\left(\bra2_{\rG}\otimes I_{\rD}\right)
   R_{\delta,\eta}\left(H/\alpha_H\right)
   \left(\bra0_{\rA}\otimes I_{\rG\rD}\right)\ket\Psi,\\
 \ket{y_0}_{\rD}
 &=\left(\bra2_{\rG}\otimes I_{\rD}\right)P
   \left(\bra0_{\rA}\otimes I_{\rG\rD}\right)\ket\Psi
   =\frac{\beta}{\kappa\gamma}\ket{AD^{-1}b}_{\rD}.
\end{aligned}
\]
The filter estimate and the fact that $\bra2_{\rG}\otimes I_{\rD}$ is a contraction imply
$\norm{y-y_0}\le\eta$ and $\norm{y_0}\le1$. The correction identity gives
\[
 \frac{I+\kappa^{-2}A^{-2}}4\,y_0=\lambda x,
 \qquad \lambda=\frac{\beta s}{4\kappa\gamma}\ge\frac1{128},
\]
where the last inequality uses $\gamma\le s/\kappa$. Consequently,
\begin{align}
 \norm{z-\lambda x}
 &\le\norm{c_{\kappa,\eta}(A)\left(y-y_0\right)}
 +\left\|\left(c_{\kappa,\eta}(A)
 -\frac{I+\kappa^{-2}A^{-2}}4\right)y_0\right\|\notag\\
 &\le\frac34\eta+\frac\eta2\le2\eta.\label{eq:refined-error}
\end{align}
For $z\ne0$, the reverse triangle inequality then yields
\[
 \left\|\frac z{\norm z}-x\right\|
 \le\frac{|\norm z-\lambda|+\norm{z-\lambda x}}{\lambda}
 \le\frac{4\eta}{\lambda}\le512\eta.
\]
The normalized error is at most $\varepsilon/2$, and
\[
 \norm z\ge\lambda-2\eta
 \ge\frac1{128}-\frac{\varepsilon}{512}>\frac1{256}.
\]
Thus $z\ne0$ and the success probability is greater than $1/65536$. These estimates concern joint acceptance, without conditioning on any earlier preparation measurement.

\emph{Resources.}
Each polynomial block uses $O(\kappa\log(1/\eta))$ matrix queries and $O(\kappa\log(1/\eta)(a+1))$ extra gates, with no vector queries. Since $\eta=\varepsilon/1024$, $\log(1/\eta)=\Theta(\log(1/\varepsilon))$. The joint test adds $O(a+\log\kappa+1)$ gates for the preparation and signal registers, which is absorbed in the bound.
\end{proof}

\subsection{Proof of the main theorem}\label{sec:upper-proof}

\begin{theorem}[Simultaneously optimal oracle complexity]\label{thm:upper}
Let $A\in\C^{d\times d}$ be invertible, and $b\in\C^d$ be a unit vector.
There is a quantum algorithm that, given query access to an exact
$(\alpha,a,0)$-block-encoding $U_A$ of $A$, a state-preparation unitary $U_b$
preparing $\ket b$, a normalized condition-number bound $\kappa\ge2$ satisfying
$\alpha\norm{A^{-1}}\le\kappa$, an estimate $\widehat s$ satisfying
$\alpha\norm{A^{-1}b}/2\le\widehat s\le
2\alpha\norm{A^{-1}b}$, and $0<\varepsilon<1/2$, outputs
$\ket{\widetilde x}$ with probability at least $2/3$ such that
\[
 \left\|\widetilde x-
 \frac{A^{-1}b}{\norm{A^{-1}b}}\right\|\le\varepsilon.
\]
Moreover, the algorithm uses $O(\kappa\log(1/\varepsilon))$ queries to $U_A$,
and $O(\kappa/(\alpha\norm{A^{-1}b}))$ queries to $U_b$, and
$O(\kappa(a+1)\log(1/\varepsilon)
+\kappa(1+\log d)/(\alpha\norm{A^{-1}b}))$ extra one- and two-qubit gates.
It uses $O(\log d+a+\log\kappa)$ qubits besides the oracle workspaces.
The same conclusions hold under the relaxed estimate promise
$3\alpha\norm{A^{-1}b}/8\le\widehat s\le
5\alpha\norm{A^{-1}b}/2$.
\end{theorem}

\begin{proof}
For a normalized Hermitian input, \cref{prop:coarse,prop:refinement} show that
one run has acceptance probability greater than
$1/65536$ and, conditioned on acceptance, outputs a state within distance
$\varepsilon/2$ of $\ket x$.

Each run uses the same circuit and accepting projection. Since
$72000/65536>\log3$, the probability of failing in all $72000$ runs is less
than $e^{-72000/65536}<1/3$. Every successful run has the same pure output,
so returning the first one preserves the conditional state guarantee.

For one run, $\textup{\Coarse}$ uses $O(\kappa)$ matrix queries and
$O(\kappa/s)$ vector queries, whereas
$\textup{\FilterKernel}$ and $\textup{\CorrectSolution}$ together use
$O(\kappa\log(1/\varepsilon))$ matrix queries and no vector
queries. Since the number of runs is constant, these are the worst-case
query bounds in \cref{thm:upper}.

The extra-gate cost per run is
\[
 O\!\left(\kappa\left(a+1\right)+\frac\kappa s\left(n+1\right)
          +\kappa\left(a+1\right)\log\frac1\varepsilon\right),
\]
which gives the extra-gate bound in \cref{thm:upper}. The acceptance test and
reset costs are absorbed in this bound.

The registers $\rD,\rG,\rA,\rF,\rC$ use
$O(n+a+\log\kappa)$ qubits besides the oracle workspaces; the polynomial phase
lists are classical instructions.
\end{proof}

\section{Simultaneous query lower bounds}\label{sec:lower}

For an algorithm $\mathcal A$ and an input $I$, let
$Q_A(\mathcal A,I)$ and $Q_b(\mathcal A,I)$ be the largest numbers of
matrix and vector queries, respectively, in any run on $I$.  We show that one
input can make both quantities large.  The proof first selects a matrix-hard
member of a computation-history family.  A fixed least-eigenvalue direction
then makes every member of that family vector-query hard, so the selected
input has both properties.

\begin{theorem}[Same-instance query lower bounds]
\label[theorem]{thm:lower}
There are absolute constants $\kappa_0\ge4$ and $\varepsilon_0>0$ with the
following property.  Fix $\kappa\ge\kappa_0$,
$0<\varepsilon\le\varepsilon_0$, and $1\le\widehat s\le\kappa$.
Let $\mathcal A$ be any algorithm that, given query access to an exact
$(1,a,0)$-block-encoding $U_A$ of $A$ and to a unitary $U_b$ satisfying
$U_b\ket0=\ket b$, solves every Hermitian input with
\[
 \norm A\le1,\qquad \norm{A^{-1}}\le\kappa,\qquad
 \norm b=1,\qquad
 \frac{\norm{A^{-1}b}}2\le\widehat s\le2\norm{A^{-1}b},
\]
where $a=O(1)$.  The algorithm succeeds with probability at least $2/3$
and, conditioned on success, its output $\rho$ satisfies
\[
 \Dtr(\rho,\proj x)\le\varepsilon,
 \qquad x=\frac{A^{-1}b}{\norm{A^{-1}b}}.
\]
There is a single such input $I_*$ for which
\[
 Q_A(\mathcal A,I_*)
   =\Omega\!\left(\kappa\log\frac1\varepsilon\right),
 \qquad
 Q_b(\mathcal A,I_*)=\Omega\!\left(\frac\kappa{\widehat s}\right).
\]
The matrix of $I_*$ can be chosen real Hermitian with
$\norm A=1$, $\norm{A^{-1}}=\kappa$, dimension
$O(\kappa\log(1/\varepsilon))$, and an exact block-encoding with a constant
number of signal qubits.
\end{theorem}

\subsection{A norm-controlled computation history}

The construction below specializes the cyclic computation-history reduction
of \cite[Sections 6--7]{Patel26} to reversible parity computation; the
underlying history-state reduction goes back to \cite{HHL09}.  We include the
calculation to show that the solution norm is independent of $z$ and lies
between $\kappa/2$ and $\kappa$, while the clock interval storing the parity
has probability at least $\lambda^{2m}/256$ in the normalized solution.

\begin{definition}[Cyclic parity-history family]
\label[definition]{def:parity-history-family}
Fix $\kappa\ge4$, $m\ge1$, and $z=(z_1,\ldots,z_m)\in\{0,1\}^m$.  Set
\begin{equation}\label{eq:lower-history-parameters}
 \lambda=\frac{\kappa-1}{\kappa+1},\qquad
 \ell=\lceil8\kappa\rceil,\qquad
 L=2\ell+2m.
\end{equation}
Let $\mathsf K\simeq\C^L$ be a clock register with computational basis
$\{\ket0,\ldots,\ket{L-1}\}$, and let $\mathsf W\simeq\C^2$ be a one-qubit
work register.  On $\mathsf W$, use the following sequence of $L$ transition
gates:
\begin{equation}\label{eq:lower-history-transitions}
 \underbrace{I,\ldots,I}_{\ell-1},\quad
 X^{z_1},\ldots,X^{z_m},\quad
 \underbrace{I,\ldots,I}_{\ell},\quad
 X^{z_m},\ldots,X^{z_1},\quad I.
\end{equation}
Writing these gates as $V_0,\ldots,V_{L-1}$, define the following operators
and state on $\mathsf K\otimes\mathsf W$:
\begin{equation}\label{eq:lower-history-matrix}
 B_z=\sum_{j=0}^{L-1}\ket{j+1\bmod L}\bra j\otimes V_j,
 \qquad
 H_z=\frac{I-\lambda B_z}{1+\lambda},
 \qquad
 \ket{b_{\rm hist}}=
 \frac1{\sqrt\ell}\sum_{i=0}^{\ell-1}
 \ket i_{\mathsf K}\ket0_{\mathsf W}.
\end{equation}
The \emph{cyclic parity-history family} is
\[
 \mathcal H_{\kappa,m}
 :=\bigl\{(H_z,b_{\rm hist}):z\in\{0,1\}^m\bigr\}.
\]
\end{definition}

\begin{lemma}
\label[lemma]{lem:history-profile}
Every $(H_z,b_{\rm hist})\in\mathcal H_{\kappa,m}$ from
\cref{def:parity-history-family} satisfies
\[
 \norm{H_z}=1,\qquad \norm{H_z^{-1}}=\kappa,\qquad
 \frac\kappa2\le
 Y:=\norm{H_z^{-1}b_{\rm hist}}\le\kappa.
\]
The number $Y$ is independent of $z$.  Define the normalized history
solution by
\[
 \ket{\psi_z}:=\frac{H_z^{-1}\ket{b_{\rm hist}}}{Y}.
\]
If register $\mathsf K$ of $\ket{\psi_z}$ is measured, the outcome lies in
\[
 \mathcal T=\{\ell+m-1,\ldots,2\ell+m-2\}
\]
with probability at least
\begin{equation}\label{eq:lower-history-tail}
 \frac1{256}\lambda^{2m}.
\end{equation}
Conditioned on any clock outcome in $\mathcal T$, register $\mathsf W$ is the
computational-basis state $\ket{z_1\oplus\cdots\oplus z_m}$.
\end{lemma}

\begin{proof}
Let $G_0=I$, $G_{j+1}=V_jG_j$, and
$D_z=\sum_j\proj j\otimes G_j$.  The product of the gates in
\eqref{eq:lower-history-transitions} is the identity, so
\[
 B_z=D_z(C_L\otimes I)D_z^\dagger,
 \qquad C_L\ket j=\ket{j+1\bmod L}.
\]
Since $L$ is even, the spectrum of $C_L$ contains both $1$ and $-1$.
The singular values of $I-\lambda B_z$ therefore range from
$1-\lambda$ to $1+\lambda$, proving the two operator-norm identities.

The first $\ell-1$ transitions are identities, hence
$D_z^\dagger b_{\rm hist}=b_{\rm hist}$.  Expanding the inverse around the
cycle gives
\begin{equation}\label{eq:lower-history-amplitudes}
 H_z^{-1}\ket{b_{\rm hist}}=D_z(\ket y\ket0),
 \qquad
 y_j=\frac{1+\lambda}{\sqrt\ell(1-\lambda^L)}
      \sum_{i=0}^{\ell-1}\lambda^{(j-i)\bmod L}.
\end{equation}
Thus $Y=\norm y$ is independent of $z$, and $Y\le\kappa$ follows from
$\norm{H_z^{-1}}=\kappa$.  For
$\lfloor\ell/2\rfloor\le j\le\ell-1$, retaining the terms $i\le j$ in
\eqref{eq:lower-history-amplitudes} yields
\[
 y_j\ge
 \frac\kappa{\sqrt\ell}\left(1-\lambda^{j+1}\right)
 \ge\frac\kappa{\sqrt\ell}\left(1-\lambda^{\ell/2}\right).
\]
Moreover,
$\lambda^{\ell/2}\le e^{-\ell/(\kappa+1)}<1/4$.  Summing the squares
over at least $\ell/2$ indices gives $Y^2\ge9\kappa^2/32$, and hence
$Y\ge\kappa/2$.

After the $m$ input-dependent transitions, the next $\ell$ clock values form
$\mathcal T$ and the work bit stores parity.  For $j\ge\ell-1$, there is no
wraparound in \eqref{eq:lower-history-amplitudes}, so
\[
 y_{\ell-1+k}=\lambda^k y_{\ell-1},\qquad
 y_{\ell-1}\ge
 \frac\kappa{\sqrt\ell}(1-\lambda^\ell).
\]
Using $Y\le\kappa$,
\begin{align*}
 \frac{\sum_{j\in\mathcal T}y_j^2}{Y^2}
 &\ge
 \frac{(1-\lambda^\ell)^2}{\ell}\lambda^{2m}
 \frac{1-\lambda^{2\ell}}{1-\lambda^2}.
\end{align*}
Here $\ell\le9\kappa$, $1-\lambda^2\le4/\kappa$, and
$\lambda^\ell\le1/2$.  The coefficient of $\lambda^{2m}$ is therefore at
least $1/256$, proving \eqref{eq:lower-history-tail}.
\end{proof}

Let the hidden string be accessed through the bit oracle
$\ket i\ket c\mapsto\ket i\ket{c\oplus z_i}$.

\begin{proposition}[Norm-controlled history family]
\label[proposition]{prop:history-family}
For every $\kappa\ge4$, $1\le\widehat s\le\kappa$, and $m\ge1$, there is a
family of inputs
\[
 I_z=(A_z,b,U_{A,z},U_b),\qquad z\in\{0,1\}^m,
\]
with the following properties.
\begin{enumerate}
\item Each $A_z$ is real Hermitian and
\[
 \norm{A_z}=1,\qquad \norm{A_z^{-1}}=\kappa,\qquad
 s_*:=\norm{A_z^{-1}b}\in[\widehat s/2,3\widehat s/2].
\]
The vector $b$, its full preparation oracle $U_b$, and $s_*$ are independent
of $z$.
\item The unitary $U_{A,z}$ is an exact block-encoding with a constant number
of signal qubits.  Every matrix query, including an adjoint or controlled
query, can be simulated with $O(1)$ queries to the bits of $z$.
\item There is a fixed unit vector $e$, orthogonal to both $b$ and
$A_z^{-1}b$, such that $A_ze=\kappa^{-1}e$ for every $z$.
\item For $x_z=A_z^{-1}b/s_*$, there is a fixed Hermitian observable $M$,
$\norm M\le1$, such that
\begin{equation}\label{eq:lower-family-signal}
 (-1)^{z_1+\cdots+z_m}\bra{x_z}M\ket{x_z}
 \ge c\lambda^{2m}
\end{equation}
for an absolute constant $c>0$.
\end{enumerate}
Each matrix $A_z$ has dimension
$4L+4=8\ell+8m+4=O(\kappa+m)$.
\end{proposition}

\begin{proof}
\emph{Norm adjustment.}
Let $Y$ be as in \cref{lem:history-profile} and set
\[
 s_*=\min\{3\widehat s/2,Y\}.
\]
Since $Y\ge\kappa/2\ge\widehat s/2$ and $Y\ge2$, we have
$s_*\in[\widehat s/2,3\widehat s/2]$ and $s_*\ge3/2$.  Define
\[
 G_z=1\oplus H_z\oplus\kappa^{-1},
\]
where $\ket u$ is the unit vector spanning the first
one-dimensional block, on which $G_z$ acts as the identity.  This block is
used only to adjust the solution norm.  Define the vector
\begin{equation}\label{eq:lower-mixed-source}
 \ket g=
 \sqrt{\frac{Y^2-s_*^2}{Y^2-1}}\ket u
 \oplus
 \sqrt{\frac{s_*^2-1}{Y^2-1}}\ket{b_{\rm hist}}
 \oplus0.
\end{equation}
Here the final zero lies in the last one-dimensional block, on which $G_z$
acts as multiplication by $\kappa^{-1}$.
The definitions give
\[
 \norm g^2
 =\frac{Y^2-s_*^2}{Y^2-1}+\frac{s_*^2-1}{Y^2-1}=1,
 \qquad
 \norm{G_z^{-1}g}^2
 =\frac{Y^2-s_*^2}{Y^2-1}
  +Y^2\frac{s_*^2-1}{Y^2-1}
 =s_*^2.
\]
The $H_z$-block of the normalized vector $G_z^{-1}g/s_*$ has probability
\begin{equation}\label{eq:lower-history-mass}
 \frac{(s_*^2-1)Y^2}{s_*^2(Y^2-1)}
 =\left(1-\frac1{s_*^2}\right)\frac{Y^2}{Y^2-1}
 \ge\frac59.
\end{equation}

\emph{Hermitian embedding.}
Set
\[
 A_z=\begin{pmatrix}0&G_z\\G_z^\dagger&0\end{pmatrix},
 \qquad
 \ket b=\ket g\oplus0.
\]
Then $A_z^{-1}b=0\oplus G_z^{-1}g$, so the embedding preserves the
solution norm.  It also gives $\norm{A_z}=1$ and
$\norm{A_z^{-1}}=\kappa$.  The scalar $\kappa^{-1}$ block of $G_z$
produces a fixed eigenvector $e$ of $A_z$ with eigenvalue $\kappa^{-1}$.
Because the last component of $g$ is zero, this vector is orthogonal to both
$b$ and $A_z^{-1}b$.

On the solution half and the history block, let $M$ apply Pauli $Z$ to the
work qubit when the clock lies in $\mathcal T$, and let it be zero on every
other subspace.  Combining \eqref{eq:lower-history-tail} and
\eqref{eq:lower-history-mass} proves
\eqref{eq:lower-family-signal}, for example with $c=5/2304$.

\emph{Oracle implementation.}
For each clock transition, reversibly compute whether $V_j=X^{z_i}$ and, if
so, the index $i$.  Two bit queries then implement $V_j$ coherently and
uncompute the queried bit.  Hence $B_z$, $B_z^\dagger$, and their controlled
versions use $O(1)$ bit queries.  Applying the standard two-term LCU
construction to $H_z=(I-\lambda B_z)/(1+\lambda)$, followed by the fixed
direct sums and Hermitian embedding above, gives an exact block-encoding
$U_{A,z}$ with $O(1)$ signal qubits and $O(1)$ bit queries \cite{GSLW19}.
All LCU coefficients and the preparation of $g$ are independent of $z$.

Finally, $H_z$ acts on $\mathsf K\otimes\mathsf W$ and therefore has
dimension $2L$.  Thus $G_z$ has dimension $2L+2$, and its Hermitian embedding
$A_z$ has dimension $4L+4=8\ell+8m+4=O(\kappa+m)$.  
\end{proof}

\subsection{Pointwise vector-query hardness}

The fixed vector $e$ in \cref{prop:history-family} makes every member of the
family hard for the state-preparation oracle. 

\begin{lemma}[Pointwise vector-query lower bound]
\label[lemma]{lem:pointwise-vector}
There is an absolute constant $\varepsilon_v>0$ such that the following
holds.  Let $I_z$ be any input from \cref{prop:history-family}, and let
$\mathcal A$ solve every input in the promise class of
\cref{thm:lower} with conditional trace-distance error at most
$\varepsilon_v$.  Then
\[
 Q_b(\mathcal A,I_z)=\Omega(\kappa/\widehat s),
\]
where only vector queries are counted.
\end{lemma}

\begin{proof}
Let $e$ be the fixed eigenvector in \cref{prop:history-family}, put
$\tau=5\widehat s/(4\kappa)$, and replace $b$ by
\[
 \ket{b'}=\frac{\ket b+\tau\ket e}{\sqrt{1+\tau^2}}.
\]
Keep $A_z$, $\kappa$, and $\widehat s$ unchanged.  The new solution norm is
\[
 s'=
 \frac{\sqrt{s_*^2+25\widehat s^2/16}}
      {\sqrt{1+25\widehat s^2/(16\kappa^2)}}.
\]
Since $s_*\in[\widehat s/2,3\widehat s/2]$ and $\widehat s\le\kappa$,
\begin{equation}\label{eq:lower-paired-norm}
 \sqrt{\frac{29}{41}}\,\widehat s
 \le s'\le\frac{\sqrt{61}}4\,\widehat s<2\widehat s.
\end{equation}
Thus $\widehat s$ is also a factor-two estimate of $s'$.

The two normalized solutions obey
\[
 \ket{x_z'}=
 \frac{s_*\ket{x_z}+(5\widehat s/4)\ket e}
      {\sqrt{s_*^2+25\widehat s^2/16}},
 \qquad
 d_0:=\Dtr(\proj{x_z},\proj{x_z'})
 \ge\frac5{\sqrt{61}}>\frac58.
\]
Let $\theta=\arctan\tau$, let $R$ be the rotation by $\theta$ in
$\operatorname{span}\{b,e\}$ taking $b$ to $b'$, and set
$U_{b'}=RU_b$.  Then, for $U_b,U_b^\dagger$ and their controlled versions,
the corresponding query gates for the two inputs differ in operator norm by
at most
\begin{equation}\label{eq:lower-vector-oracle-distance}
 \delta:=\norm{I-R}=2\sin(\theta/2)\le\theta\le\tau
 =O(\widehat s/\kappa).
\end{equation}

It remains to make the hybrid argument pointwise in the original input.
Write
\[
 q:=Q_b(\mathcal A,I_z).
\]
Define a truncated algorithm $\mathcal A_{\le q}$ that simulates
$\mathcal A$, but halts with a separate abort flag whenever
$\mathcal A$ is about to make its $(q+1)$st vector query.  On $I_z$ the abort
flag is never produced, by the definition of $q$, so
$\mathcal A_{\le q}$ has exactly the same final flagged state as
$\mathcal A$ on $I_z$.  On either oracle, $\mathcal A_{\le q}$ makes at most
$q$ vector queries.  Purifying its internal measurements and applying the
standard telescoping hybrid argument therefore gives
\begin{equation}\label{eq:lower-truncated-hybrid}
 \Dtr(\Omega,\widetilde\Omega')\le q\delta,
\end{equation}
where $\Omega$ is the final flagged state on $I_z$ and
$\widetilde\Omega'$ is the truncated final flagged state on the perturbed
input $I_z'=(A_z,b',U_{A,z},U_{b'})$.

Let $r$ be the abort probability of $\mathcal A_{\le q}$ on $I_z'$.
Since the abort probability is zero on $I_z$, data processing applied to
\eqref{eq:lower-truncated-hybrid} gives $r\le q\delta$.  Let $\Omega'$ be
the final flagged state of the untruncated algorithm $\mathcal A$ on $I_z'$.
The truncated and untruncated executions agree unless the abort event occurs,
so
\[
 \Dtr(\widetilde\Omega',\Omega')\le r\le q\delta.
\]
Consequently,
\begin{equation}\label{eq:lower-one-sided-hybrid}
 \Dtr(\Omega,\Omega')\le2q\delta.
\end{equation}
This is the step that makes the lower bound depend on the query count on the
\emph{original} input $I_z$, rather than only on the maximum query count over
the pair $I_z,I_z'$; compare the stopping-time argument in
\cite[Theorem 1]{SommaSubasi21}.

Finally, let $p,p'\ge2/3$ be the success probabilities of $\mathcal A$ on
$I_z$ and $I_z'$, and let $\rho,\rho'$ be the corresponding conditional
successful outputs.  The solver guarantee and the triangle inequality give
\[
 \Dtr(\rho,\rho')\ge d_0-2\varepsilon_v.
\]
Because success and failure occupy orthogonal flag subspaces, the full flagged
states satisfy
\[
 \Dtr(\Omega,\Omega')
 \ge \min\{p,p'\}\,\Dtr(\rho,\rho')
 \ge \frac23\left(d_0-2\varepsilon_v\right).
\]
Choose, for example, $\varepsilon_v\le5/(4\sqrt{61})\le d_0/4$.  Combining
this constant lower bound with \eqref{eq:lower-one-sided-hybrid} and
\eqref{eq:lower-vector-oracle-distance} yields
\[
 q=\Omega(1/\delta)
  =\Omega(\kappa/\widehat s).
\]
Thus every original family member $I_z$ is vector-query hard, as required.
\end{proof}

\subsection{Selecting a single hard instance}

We choose the history length so that the solver error cannot change the sign
of the parity-dependent expectation.  A reduction to parity then selects a
member of the family that is matrix-query hard.  Since
\cref{lem:pointwise-vector} makes
every member vector-query hard, the selected member satisfies both lower
bounds.

\begin{lemma}[Unbounded-error parity lower bound]
\label[lemma]{lem:parity-lower}
Let $\mathcal C$ be a quantum algorithm with query access to
$z\in\{0,1\}^m$.  If $\mathcal C$ outputs $z_1\oplus\cdots\oplus z_m$ with
probability strictly greater than $1/2$ for every $z$, then some run of
$\mathcal C$ uses at least $m/2$ bit queries.
\end{lemma}

\begin{proof}
Let $P(z)$ be the probability that $\mathcal C$ outputs $0$ minus the
probability that it outputs $1$.  If every run uses at most $T$ queries, the
polynomial method gives a real multilinear polynomial $P$ of degree at most
$2T$ \cite{Beals01}.  The success assumption gives
\[
 (-1)^{z_1+\cdots+z_m}P(z)>0
 \qquad\text{for every }z.
\]
If $\deg P<m$, then each monomial omits some variable, and summing over that
variable gives
\[
 2^{-m}\sum_z(-1)^{z_1+\cdots+z_m}P(z)=0.
\]
This contradicts the strict positivity above.  Hence $2T\ge m$.
\end{proof}

\begin{proof}[Proof of \cref{thm:lower}]
\emph{Preserving the parity-dependent expectation.}
Choose a sufficiently small absolute constant $c_0>0$ and set
\[
 m=\left\lfloor c_0\kappa\log\frac1\varepsilon\right\rfloor.
\]
After fixing $\kappa_0$ and $\varepsilon_0$, we have
$m=\Theta(\kappa\log(1/\varepsilon))$ and $m\ge1$.  Moreover,
$\log(1/\lambda)<3/\kappa$ for $\kappa\ge4$, so the parity-dependent
expectation in \eqref{eq:lower-family-signal} is larger than the solver error:
\[
 c\lambda^{2m}\ge c e^{-6m/\kappa}
 \ge c\varepsilon^{6c_0}>4\varepsilon,
\]
where the last inequality follows by decreasing $c_0$ and
$\varepsilon_0$ if necessary.

\emph{Selecting a matrix-query hard member.}
Apply \cref{prop:history-family} with this $m$ and fix a solver
$\mathcal A$.  Let
\[
q=\max_z Q_A(\mathcal A,I_z).
\]
Use $\mathcal A$ to compute the parity of $z$.  If the solver succeeds,
apply the two-outcome measurement $\{(I+M)/2,(I-M)/2\}$ and output $0$ or
$1$, respectively; if it fails, output a uniformly random bit.  Let
$p_z\ge2/3$ be the solver's success probability and $\rho_z$ its conditional
output.  The difference between the probabilities of outputting $0$ and $1$
is
\[
 B(z)=p_z\tr(M\rho_z).
\]
The trace-distance guarantee changes the expectation of $M$ by at most
$2\varepsilon$.  Together with \eqref{eq:lower-family-signal} and the choice
of $m$, this gives
\[
 (-1)^{z_1+\cdots+z_m}B(z)
 \ge p_z\bigl(c\lambda^{2m}-2\varepsilon\bigr)>0
 \qquad\text{for all }z.
\]
Thus the parity algorithm is correct with probability strictly greater than
$1/2$ for every $z$.  The only part of $I_z$ that depends on $z$ is
$U_{A,z}$, and each matrix query uses $O(1)$ bit queries by
\cref{prop:history-family}.  The parity algorithm therefore uses $O(q)$ bit
queries.  By \cref{lem:parity-lower}, $q=\Omega(m)$, so some $z_*$ satisfies
\[
 Q_A(\mathcal A,I_{z_*})
 =\Omega\!\left(\kappa\log\frac1\varepsilon\right).
\]

\emph{Obtaining both bounds on the same member.}
Choose $\varepsilon_0\le\varepsilon_v$.  Applying
\cref{lem:pointwise-vector} to this same input gives
$Q_b(\mathcal A,I_{z_*})=\Omega(\kappa/\widehat s)$.  The dimension and oracle
claims follow from \cref{prop:history-family}.
\end{proof}

\section{Approximate block-encodings}\label{sec:robustness}

We return to the general normalization in \cref{prob:qlsa}. A fixed approximate block-encoding is an exact encoding of its compressed matrix. We apply the solver to that matrix and compare the two normalized solutions using matrix perturbation bounds \cite[Chapter 7]{Higham02}. It remains to check the inverse-norm bound and transfer the supplied norm estimate.

\subsection{Perturbation of the normalized solution}

\begin{lemma}[Inverse, solution norm, and normalized solution]\label[lemma]{lem:backward}
Let $A,B\in\C^{d\times d}$, let $A$ be invertible, and let $b\in\C^d$ be nonzero. If $0\le\rho<1$ and $\norm{A^{-1}}\norm{B-A}\le\rho$, then $B$ is invertible and
\begin{align}
 \norm{B^{-1}}&\le\frac{\norm{A^{-1}}}{1-\rho},\label{eq:inverse-perturb}\\
 \frac{\norm{A^{-1}b}}{1+\rho}
 &\le\norm{B^{-1}b}\le\frac{\norm{A^{-1}b}}{1-\rho},\label{eq:norm-perturb}\\
 \left\|\frac{B^{-1}b}{\norm{B^{-1}b}}
       -\frac{A^{-1}b}{\norm{A^{-1}b}}\right\|&\le2\rho.
 \label{eq:state-perturb}
\end{align}
\end{lemma}

\begin{proof}
Set $E=A^{-1}(B-A)$, so $\norm E\le\rho<1$ and $B=A(I+E)$. The Neumann series gives
\[
 \left(I+E\right)^{-1}=\sum_{j=0}^{\infty}\left(-E\right)^j,
 \qquad \norm{\left(I+E\right)^{-1}}\le\frac1{1-\rho}.
\]
Thus $B^{-1}=(I+E)^{-1}A^{-1}$, proving \eqref{eq:inverse-perturb}.

Let $y=A^{-1}b$ and $z=B^{-1}b$. From $Bz=Ay$ we have
$y=(I+E)z$, and hence
\[
 \norm{y-z}\le\rho\norm z,\qquad
 \left(1-\rho\right)\norm z\le\norm y\le\left(1+\rho\right)\norm z.
\]
The latter inequalities prove \eqref{eq:norm-perturb}. Both vectors are nonzero. For their normalized versions,
\begin{align*}
 \left\|\frac y{\norm y}-\frac z{\norm z}\right\|
 &\le\left\|\frac y{\norm y}-\frac y{\norm z}\right\|
       +\frac{\norm{y-z}}{\norm z}\\
 &=\frac{|\norm z-\norm y|+\norm{y-z}}{\norm z}
 \le\frac{2\norm{y-z}}{\norm z}\le2\rho.
\end{align*}
Using $\norm z$ in the denominator is what avoids a factor $(1-\rho)^{-1}$ in \eqref{eq:state-perturb}.
\end{proof}

\subsection{The algorithm for a fixed approximate encoding}

\begin{theorem}[Fixed approximate block-encodings]\label{thm:robust}
Let $A\in\C^{d\times d}$ be invertible with $\norm A\le\alpha$, and let
$b\in\C^d$ be a unit vector. There is a quantum algorithm that, given query
access to a fixed $(\alpha,a,\delta_A)$-block-encoding $U_A$ of $A$, a
state-preparation unitary $U_b$ preparing $\ket b$, a normalized
condition-number bound $\kappa\ge2$ satisfying
$\alpha\norm{A^{-1}}\le\kappa$, an estimate $\widehat s$ satisfying
$\alpha\norm{A^{-1}b}/2\le\widehat s\le
2\alpha\norm{A^{-1}b}$, and $0<\varepsilon<1/2$, has the following
guarantee. If $\kappa\delta_A/\alpha\le1/4$, it outputs
$\ket{\widetilde x}$ with probability at least $2/3$ such that
\begin{equation}\label{eq:robust-error}
 \left\|\widetilde x-
 \frac{A^{-1}b}{\norm{A^{-1}b}}\right\|
 \le\varepsilon+\frac{2\kappa\delta_A}{\alpha}.
\end{equation}
 The algorithm
uses $O(\kappa\log(1/\varepsilon))$ queries to $U_A$,
$O(\kappa/(\alpha\norm{A^{-1}b}))$ queries to $U_b$, and
\[
 O\!\left(\kappa(a+1)\log\frac1\varepsilon
 +\frac{\kappa(1+\log d)}{\alpha\norm{A^{-1}b}}\right)
\]
extra one- and two-qubit gates. 
\end{theorem}

\begin{proof}
Put $\rho=\kappa\delta_A/\alpha$ and
$s=\alpha\norm{A^{-1}b}$. Define the actual encoded matrix and its solution
norm by
\[
 B=\alpha\left(\bra{0^a}\otimes I\right)U_A\left(\ket{0^a}\otimes I\right),
 \qquad s_B=\alpha\norm{B^{-1}b}.
\]
The first identity implies $\norm B\le\alpha$. By the given approximation promise,
$\norm{B-A}\le\delta_A$, so
$\norm{A^{-1}}\norm{B-A}\le\kappa\delta_A/\alpha=\rho\le1/4$.
\Cref{lem:backward} shows that $B$ is invertible, and gives
\begin{equation}\label{eq:robust-promises}
 \alpha\norm{B^{-1}}\le\frac\kappa{1-\rho}\le\frac{4\kappa}3,
 \qquad \frac{s}{1+\rho}\le s_B\le\frac{s}{1-\rho}.
\end{equation}
The supplied estimate satisfies the relaxed condition of \cref{thm:upper}
for $B$:
\[
 \frac{\widehat s}{s_B}
 \ge\frac{s/2}{s/\left(1-\rho\right)}=\frac{1-\rho}2\ge\frac38,
 \qquad
 \frac{\widehat s}{s_B}
 \le\frac{2s}{s/\left(1+\rho\right)}=2\left(1+\rho\right)\le\frac52.
\]
The unitary $U_A$ is an exact $(\alpha,a,0)$-block-encoding of $B$.
Therefore, \eqref{eq:robust-promises} and the estimate bounds above allow us
to apply \cref{thm:upper} to $B$ with inverse-norm bound $4\kappa/3$ and the
same estimate $\widehat s$. It succeeds with probability at least $2/3$ and
outputs a state within $\varepsilon$ of
$B^{-1}b/\norm{B^{-1}b}$. Applying \eqref{eq:state-perturb} and the triangle
inequality yields
\[
 \left\|\widetilde x-\frac{A^{-1}b}{\norm{A^{-1}b}}\right\|
 \le\left\|\widetilde x-\frac{B^{-1}b}{\norm{B^{-1}b}}\right\|
   +\left\|\frac{B^{-1}b}{\norm{B^{-1}b}}
            -\frac{A^{-1}b}{\norm{A^{-1}b}}\right\|
 \le\varepsilon+2\rho.
\]
For query counts, the matrix term is
$O((4\kappa/3)\log(1/\varepsilon))=O(\kappa\log(1/\varepsilon))$. The vector query is
\[
 \frac{4\kappa}{3s_B}\le\frac{4\left(1+\rho\right)}3\frac\kappa s
 \le\frac53\frac\kappa s. \qedhere
\]
\end{proof}

\appendix
\section{Finite implementation of the preparation transducer}\label[appendix]{app:finite}

The general compiler and its error bound are imported from
\cite[Theorem 7.1]{BJY24}; we do not repeat its implementation proof.  This
appendix records the two facts needed by \cref{sec:graph}: the compiler can
implement its clock and control operations with real gates, and the catalyst
bounds in \cref{prop:preparation-transducer} give the required separate oracle
costs.

\subsection{Real clock and control gates}

\begin{lemma}[Real clock and control gates]\label[lemma]{lem:real-compiler}
Let $S$ be a canonical transducer with two input oracles, and let
$K,K_1,K_2$ be powers of two with $1\le K_i\le K$.  The implementation in
\cite[Theorem 7.1]{BJY24} can be chosen so that every compiler gate other
than the controlled $S^\circ$ circuits and oracle queries is real in the
computational basis.  These gates prepare and unprepare the clock and perform
the label-conditioned clock updates.  The implementation uses
$O(K+K_1+K_2)$ extra one- and two-qubit gates and $O(1+\log K)$ clock and
ancilla qubits.  All ancilla qubits are returned exactly to $\ket0$.
\end{lemma}

\begin{proof}
The cited construction prepares and unprepares a uniform clock with
Hadamards and otherwise uses controlled $S^\circ$ circuits, oracle queries,
and updates of the clock value conditioned on the label register.
The construction in \cite[Lemma 4.6]{BJY24} implements each conditional update
$\ket t\mapsto\ket{t+1\bmod K}$ with real reversible gates and returns its work
qubits to $\ket0$.  For oracle $i$, write $D_i=K/K_i$.  The corresponding update
is $\ket t\mapsto\ket{t+D_i\bmod K}$.  Because $D_i$ is a power of two, this
increments only the high-order quotient register, which has $K_i$ possible
values.  There are $K_i$ such updates, and the same construction implements
all of them using $O(K_i)$ gates.  Thus the non-query and two query tracks cost
$O(K)$, $O(K_1)$, and $O(K_2)$, respectively.  These clock updates use NOT,
CNOT, and Toffoli primitives, with constant-size one- and two-qubit
decompositions.  Thus the clock preparation and updates use only real gates,
and all ancilla qubits return to $\ket0$.
\end{proof}

\subsection{QLSA parameters}

\begin{proposition}[Finite preparation circuit]\label[proposition]{prop:finite-preparation}
Let $S$ be the preparation transducer in \eqref{eq:transducer-work}, and
suppose $3s/8\le\widehat s\le5s/2$.  Define
\begin{equation}\label{eq:finite-budgets}
 K=K_1=2^{\lceil\log_2(128\cdot10^6\kappa)\rceil},\qquad
 K_2=2^{\lceil\log_2(8\cdot10^6(1+\kappa/\widehat s))\rceil}.
\end{equation}
There is a circuit $\Finite(S;\kappa,\widehat s)$, independent of the
catalyst, such that for the ideal output $\psi$ in
\eqref{eq:preparation-transduction},
\begin{equation}\label{eq:finite-preparation-error}
 \left\|\Finite(S;\kappa,\widehat s)
       \left(\ket0_{\rA}\ket e_{\rG\rD}\right)
       -\ket0_{\rA}\ket\psi_{\rG\rD}\right\|<10^{-3}.
\end{equation}
It uses $O(\kappa)$ controlled queries to $U_H$, $O(\kappa/s)$ controlled
queries to $R_e$, and $O(\kappa(a+1))$ extra gates in addition to those
queries.  Its clock and ancilla registers use $O(\log\kappa)$ qubits, all
clock and control gates are real, and all ancilla qubits return to $\ket0$.
\end{proposition}

\begin{proof}
The bounds in \eqref{eq:transducer-cost} give
\[
 W<9\kappa,\qquad L_H<8\kappa,
 \qquad L_e\le\frac{\kappa}{8\widehat s}.
\]
Use the exhibited catalyst in \cref{prop:preparation-transducer}.  The
budgets in \eqref{eq:finite-budgets} satisfy $K_2\le K$: indeed,
$\widehat s\ge3s/8\ge3/8$ and $\kappa\ge2$ give
$8(1+\kappa/\widehat s)<64\kappa<128\kappa$ before rounding.  Apply
\cite[Theorem 7.1]{BJY24} with $O_1=U_H$ and $O_2=R_e$.  Since $K_1=K$,
its squared error is at most
\begin{equation}\label{eq:finite-error}
 4\left(\frac WK+\frac{L_e}{K_2}\right)
 \le4\left(\frac9{128}+\frac1{64}\right)10^{-6}
 =\frac{11}{32}\,10^{-6}<10^{-6}.
\end{equation}
This proves \eqref{eq:finite-preparation-error}.  Moreover,
$K=K_1=O(\kappa)$ and
$K_2=O(1+\kappa/\widehat s)=O(\kappa/s)$, using
$\widehat s\ge3s/8$ and $1\le s\le\kappa$.  The gate bound combines the
$O(a+1)$ controlled-circuit cost from
\cref{prop:preparation-transducer} with the compiler overhead in
\cref{lem:real-compiler}; that lemma also gives the remaining claims.
\end{proof}

\section*{Acknowledgements and AI-disclosure}
The authors used Large Language Models as AI-assisted research and writing
tools throughout the preparation of this manuscript.
The tools were used to help brainstorm ideas and explore proof strategies.
Portions of the manuscript text were redrafted or modified with AI assistance
across all sections.
All final mathematical claims, algorithms, proofs, citations, and wording were
reviewed, edited, and validated by the authors.
The authors assume responsibility for all content of the paper.

\bibliographystyle{alphaurl}
\bibliography{references}
\end{document}